%% file: improving_mappings-TR.tex
\documentclass[times,10pt,twocolumn,letterpaper]{article}
\usepackage{url}
\usepackage{cvpr}
\usepackage{times}
\usepackage{epsfig}
\usepackage{graphicx}
\usepackage{amsmath}
\usepackage{placeins}
\usepackage{verbatim}
\usepackage{tocloft}
\usepackage{amssymb}
\usepackage{color}
\usepackage[ruled,vlined,linesnumbered]{algorithm2e}
\input{myspace.tex}
\newcommand{\citesuppl}[1]{\cite{#1}}

\usepackage{mathabx}
\usepackage[pagebackref=false,breaklinks=true,letterpaper=true,colorlinks,bookmarks=false]{hyperref}
\cvprfinalcopy 

\input{tex_defs.tex}
\hypersetup{hidelinks,colorlinks=false,linkbordercolor={1 1 1}}
\def\cvprPaperID{1769} 

\begin{document}

\def\AS{Alexander Shekhovtsov}
\def\emails{e-mails}
\newcommand{\mytitle}{Maximum Persistency in Energy Minimization}
\def\mythanks{This work was supported by the EU project FP7-ICT-247870 NIFTi 
and the Austrian Science Fund (FWF) under the START project BIVISION, No. Y729.}
\title{\mytitle\vskip0.5cm
Research Report\vskip0.1cm}
\author{\AS\\
Graz University of Technology\\
{\tt\small shekhovtsov@icg.tugraz.at}
\thanks{\mythanks}
}
\maketitle
\setlength{\figwidth}{0.5\linewidth}


\begin{abstract}
\input{tex/abstract.tex}
\end{abstract}

\setlength\cftparskip{1pt}
\setlength\cftbeforesecskip{1pt}
\setlength\cftbeforesubsecskip{1pt}
\setlength\cftaftertoctitleskip{0pt}
\tableofcontents

\input{tex/intro.tex}
\input{tex/background.tex}
\input{tex/improving_mappings.tex}
\input{tex/unification.tex}
\input{tex/characterization.tex}
\input{tex/properties.tex}
\input{tex/po_expansion_LP.tex}
\input{tex/windowing.tex}
\input{tex/experiments-rand.tex}
\input{tex/experiments.tex}
\input{tex/conclusion.tex}
%
%
\def\bib{bib}
\par
{
\small
\bibliographystyle{ieee}
\bibliography{\bib/suppl,\bib/maxflow,\bib/max-plus-en,\bib/kiev-en,\bib/pseudo-Bool,\bib/shekhovt,\bib/vcsp,\bib/dee}
}
\end{document}

%% file: myspace.tex
\newlength{\myskip}
\makeatletter
\let\proof\@undefined                        
\let\endproof\@undefined                  
\makeatother
\usepackage{amsthm}

\renewenvironment{enumerate}%
  {\begin{list}{\arabic{enumi}.}%
     {\topsep=0in\itemsep=0in\parsep=0pt\partopsep=0in\usecounter{enumi}}%
   }{\end{list}}

\renewenvironment{itemize}%
  {\begin{list}{$\bullet$}%
     {\topsep=0in\itemsep=0pt\parsep=0pt\partopsep=0in\usecounter{itemi}}%
   }{\end{list}\addvspace{0pt}}

\floatsep = \myskip 
\SetAlgoSkip{skipmyskip}
\SetAlgoInsideSkip{}

\makeatletter
\let\corollary\@undefined
\let\c@corollary\@undefined
\let\endcorollary\@undefined
\let\definition\@undefined
\let\c@definition\@undefined
\let\enddefinition\@undefined
\let\proof\@undefined
\let\endproof\@undefined
\let\theorem\@undefined
\let\c@theorem\@undefined
\let\endtheorem\@undefined
\let\lemma\@undefined
\let\c@lemma\@undefined
\let\endlemma\@undefined
\makeatother

\newtheoremstyle{tightItalic}
  {0.5\myskip}
  {0.5\myskip}
  {}
  {}
  {\itshape}
  {.}
  { }
  {}

\newtheoremstyle{tightBf}
  {0.5\myskip}
  {0.5\myskip}
  {}
  {}
  {\bf}
  {.}
  {.5em}
  {}

\theoremstyle{tightBf}
\newtheorem{theorem}{Theorem}
\newtheorem{corollary}{Corollary}
\newtheorem{definition}{Definition}
\newtheorem{statement}{Statement}
\newtheorem{lemma}{Lemma}

\theoremstyle{tightItalic}
\newtheorem*{proof}{Proof}

\makeatletter
\renewcommand{\paragraph}{%
  \@startsection{paragraph}{4}%
  {\z@}{3.25ex \@plus 1ex \@minus .2ex}{-1em}%
  {\normalfont\normalsize\bfseries}%
}
\makeatother

%% file: tex_defs.tex
\DeclareMathOperator{\aff}{aff}

\newcommand{\argmin}{\mathop{\rm argmin}}

\DeclareMathOperator{\conv}{conv}
\DeclareMathOperator{\coni}{coni}
\DeclareMathOperator{\Null}{null}

\newcommand{\<}{\langle}
\renewcommand{\>}{\rangle}

\renewcommand{\mid}{\:|\,}

\newcommand{\Real}{\mathbb{R}}

\newcommand{\WIx}[1]{\ensuremath{\mathbb{W}_{#1}}\xspace}
\newcommand{\SIx}[1]{\ensuremath{\mathbb{S}_{#1}}\xspace}
\newcommand{\WI}{\WIx{f}}
\newcommand{\SI}{\SIx{f}}
\newcommand{\WIg}{\WIx{g}}

\newcommand{\tab}{{\hphantom{bla}}}
\newcommand{\A}{\mathcal{A}}
\newcommand{\Bool}{\mathbb{B}}

\def\O{\mathcal{O}}

\newcommand{\V}{\mathcal{V}}
\newcommand{\I}{\mathcal{I}}
\def\U{\mathcal{U}}
\newcommand{\W}{\mathcal{W}}

\renewcommand{\P}{\mathcal{P}}
\newcommand{\E}{\mathcal{E}}

\newcommand{\EE}{\tilde{\E}}
\newcommand{\N}{\mathcal{N}}

\renewcommand{\L}{\mathcal{L}}
\newcommand{\M}{\mathcal{M}}
\newcommand{\LL}{\mathcal{L}}

\newcommand{\T}{^{\mathsf T}}

\newcounter{myRomanCounter}

\newcommand{\f}{{E_f}}
\newcommand{\g}{{E_g}}

\newcommand{\energy}[1]{{E_{#1}}}

\newcommand{\K}{\mathcal{K}}

\renewcommand{\iff}{{\textrm{iff} }}

\newcommand{\AS}{Alexander Shekhovtsov}

\newlength{\figwidth}

\def\mathrlap{\mathpalette\mathrlapinternal} 
\def\mathllap{\mathpalette\mathllapinternal}
\def\mathllapinternal#1#2{\llap{$\mathsurround=0pt#1{#2}$}}
\def\mathrlapinternal#1#2{\rlap{$\mathsurround=0pt#1{#2}$}}

\def\ij{i,j}
\def\ji{j,i}
\DeclareRobustCommand{\ind}[1]{\vphantom{f}_{#1}}
\DeclareRobustCommand{\ind}[1]{(#1)}

\def\leftbb{\mathrlap{[}\hskip1.3pt[}
\def\rightbb{]\hskip1.36pt\mathllap{]}}

\def\phi{\delta}
\def\epsilon{\varepsilon}

\newcommand{\IF}{\mbox{ \rm if }}

\newcommand{\AND}{\mbox{ \rm and }}
\newcommand{\OTHERWISE}{\mbox{ \rm otherwise}}
\newcommand{\Algorithm}[1]{{Algorithm\,\ref{#1}}}

\newcommand{\Section}[1]{\S\ref{#1}}
\newcommand{\Figure}[1]{{Figure\,\ref{#1}}}
\newcommand{\Theorem}[1]{{Theorem\,\ref{#1}}}
\newcommand{\Lemma}[1]{{Lemma\,\ref{#1}}}

\newcommand{\Statement}[1]{{Statement\,\ref{#1}}}

\newcommand{\lsub}[1]{\limits_{\begin{subarray}{l}\setlength{\arraycolsep}{0.2em}#1\end{subarray}}}
\newcommand{\csub}[1]{\limits_{\begin{subarray}{c}\setlength{\arraycolsep}{0.2em}#1\end{subarray}}}


\def\maxwi{\eqref{best L-improving}\xspace}
\def\maxsi{{(\sc max-si)}\xspace}
\def\cf{\emph{cf}\onedot} 
 

%% file: tex/abstract.tex
We consider discrete pairwise energy minimization problem (weighted constraint satisfaction, max-sum labeling) and methods that identify a globally optimal partial assignment of variables. When finding a complete optimal assignment is intractable, determining optimal values for a part of variables is an interesting possibility.
Existing methods are based on different sufficient conditions. 
We propose a new sufficient condition for partial optimality which is: (1)
verifiable in polynomial time (2) invariant to reparametrization of the problem and permutation of labels and (3) includes many existing sufficient conditions as special cases. 
We pose the problem of finding the maximum optimal partial assignment identifiable by the new sufficient condition. A polynomial method is proposed which is guaranteed to assign same or larger part of variables 
than several existing approaches. The core of the method is a specially constructed linear program that identifies persistent assignments in an arbitrary multi-label setting.

%% file: tex/intro.tex
\section{Introduction}
\paragraph{Energy Minimization}
Given a graph $(\V,\E)$ and functions $f_s \colon \L_s\to\Real$ for all $s\in\V$ and $f_{st}\colon \L_s\times\L_t \to \Real$ for all $st\in\E$, where $\L_s$ are finite sets of {\em labels}, the problem is to minimize the {\em energy} 
\begin{equation}\label{the energy}
\f(x) = f_0 + \sum_{s\in \V} f_s(x_s) +\sum_{st\in \E}f_{st}(x_{s},x_{t}),
\end{equation}
over all assignments $x \in \LL=\prod_s \L_s$ (Cartesian product). Notation $st$ denotes the ordered pair $(s,t)$ for $s,t\in \V$. The general energy minimization problem is 
APX-hard. 

\paragraph{Partial Optimality}
Let $\A\subset \V$. By $x_\A$ we denote the restriction of $x$ to $\A$. An assignment $y$ with domain $\A$ is a {\em partial assignment} denoted $(\A,y)$. 
The pair $(\A,y)$ is called {\em strong optimal partial assignment} if there holds $x^*_{\A}=y$ for {\em any} minimizer $x^*$ of $\f$. And {\em weak optimal partial assignment} if {\em there exists}  a minimizer $x^*$ of $\f$ such that $x^*_{\A}=y$.
\paragraph{Related Work}
Several fundamental results identifying optimal partial assignments are obtained from the properties of linear relaxations of some discrete problems.
An optimal solution to continuous relaxation of a mixed-integer $0$-$1$ programming problem is defined by Adams \etal.~\cite{Adams:1998} to be {\em persistent} if the set of $[0,1]$ relaxed variables realizing binary values retains the same binary values in at least one integer optimum. A mixed-integer program is said to be {\em persistent} (or possess the {\em persistency} property) if {\em every} solution to its continuous relaxation is persistent. Nemhauser \& Trotter~\cite{Nemhauser-75} proved that the vertex packing problem is persistent. 
%
%
This result was later generalized to optimization of quadratic pseudo-Boolean functions (equivalent to energy minimization with two labels) by Hammer \etal~\cite{Hammer-84-roof-duality}. The relaxed problem in this case is known as the {\em roof dual}. 
{\em Strong persistency} was also proven, stating that if a variable takes the same binary value in {\em all} optimal solutions to the relaxation, 
then {\em all} optimal solutions to the original $0$-$1$ problem take this value. However, it is a rare case that a relaxation of a particular problem is persistent.
\par
Several works considered generalization of persistency to higher-order pseudo-Boolean functions. 
Adams \etal.~\cite{Adams:1998} considered a hierarchy of continuous relaxations of $0$-$1$ polynomial programming problems. Given an optimal relaxed solution, they derive sufficient conditions on the dual multipliers which ensure that the solution is persistent. This result generalizes the roof duality approach, coinciding with it in the case of quadratic polynomials in binary variables.
Kolmogorov~\cite{Kolmogorov10-bisub,Kolmogorov12-bisub} studied submodular and bisubmodular relaxations and showed that they provide a natural generalization of the quadratic pseudo-Boolean case to higher-order terms and possess the persistency property. 
Kahl and Strandmar~\cite{KahlS12} proposed a polynomial time algorithm to find the tightest submodular relaxation. 
Lu and Williams~\cite{Lu-Williams-1987}, Ishikawa~\cite{Ishikawa-11-transform} and Fix \etal~\cite{Fix-11} obtained partial optimalities via different reductions to quadratic problems and subsequent application of the roof dual.
%
%
\paragraph{Multi-label energies}
The following methods were proposed for the pairwise model~\eqref{the energy} with multi-label variables.
Kohli \etal~\cite{kohli:icml08} reduced multi-label energy to quadratic pseoudo-Boolean and applied roof dual.
%
The family of local methods known as {\em dead end elimination} (DEE), originally proposed by Desmet \etal~\cite{Desmet-92-dee}, 
uses simple sufficient conditions that consider a variable and its immediate neighbors in the graph. 
Kovtun~\cite{Kovtun03,Kovtun-10}
proposed to construct an auxiliary submodular problem whose solution provides a partial optimal assignment for the original problem.
For the Potts model it was shown that $K$ auxiliary problems 
can be solved in time $O(\log(K)F)$, where $F$ is the time to solve a single auxiliary problem~\cite{Gridchyn-13}.
%
Swoboda~\etal~\cite{Swoboda-13} proposed a method for Potts model solving a series of LP relaxations approximately and generalized it recently to general and higher-order energies~\cite{Swoboda-14}. Unlike other approaches, methods~\cite{Desmet-92-dee,Kovtun03} are not directly related to relaxation techniques. 
\paragraph{Contribution}
We observed that in many methods there is an underlying mapping 
of labelings 
$p \colon \LL \to \LL$ 
that {\em improves} the energy of any given labeling: $\f(p(x)) \leq \f(x)$. 
It follows that there exists a minimizer in the reduced search space $p(\LL)$. However, even in the case that such mapping is given, the verification of the improving property is NP-hard (see below).
We propose instead to verify that a suitable linear {\em extension} of this mapping  improves the energy of all {\em relaxed} labelings.
This constitutes a sufficient condition which is polynomial to verify. It includes sufficient conditions used in methods~\cite{Desmet-92-dee,Kovtun-10,Hammer-84-roof-duality,kohli:icml08,Swoboda-13} as special cases. 
\par
We pose the problem of finding the {\em maximum} weak/strong optimal partial assignment identifiable by the new sufficient conditions (denoted {\sc max-wi}\,/\,{\sc max-si}, respectively). 
We propose polynomial algorithms for several classes of mappings $p$, which include many of previously proposed constructions. The algorithms involve solving the LP-relaxation and an additional linear program of a comparable size.
We give a method that improves over one-against-all method of Kovtun~\cite{Kovtun-10} (including possible free choices in this method) and subsumes the method \cite{Swoboda-13}. In the case of two labels, our method reduces to known QPBO results.
Experimental verification of correctness and quantification of achieved improvement is performed on difficult random instances. Preliminary experiments with large-scale vision problems are reported in \Section{S:windowing-s}.
\par
In our previous work~\cite{shekhovtsov-11-aux_submodular} a particular map $x\mapsto (x \vee y)\wedge z$ was extended to relaxed labelings, where $\vee$ and $\wedge$ are component-wise maximum and minimum, respectively. It allowed to relate Kovtun's methods to the standard LP-relaxation and the expansion move algorithm. In the previous work~\cite{shekhovtsov-phd} a major part of the generalized approach was presented but with algorithms for a much more narrow class of mappings and without experiments.

%% file: tex/background.tex
\section{Background}

We will assume that $st\in \E \Rightarrow ts\notin \E$. Let us denote the set $\L_s\times\L_t$ as $\L_{st}$ and the pair of labels $\ind{\ij}\in \L_{st}$ as $ij$. The following set of indices is associated with the graph $(\V,\E)$ and the set of labelings:
$\I =  \{0\}\cup \{(s,i)\mid s\in \V,\ i\in \L_s\} \cup
 \{(st,ij)\mid st\in \E,\ ij\in \L_{st}\}$.
A vector $f \in \Real^\I$ has components (coordinates) $f_0$, $f_{u}\ind{l}$, $f_{st}(i,j)$  for all $u\in \V,\ l\in \L_u$, $st\in \E,\ ij\in \L_{st}$.
We further define that $f_{ts}\ind{\ji} = f_{st}\ind{\ij}$. 
Let $\EE=\E\cup\{ts\mid st\in\E\}$, the symmetric closure of $\E$. 
The {\em neighbors} of a pixel $s$ are pixels in the set $\N(s) = \{t\mid st\in \EE \}$. 
%
%
%
%
%
\nocite{Werner-PAMI07}
\begin{figure}
\includegraphics[width=\linewidth]{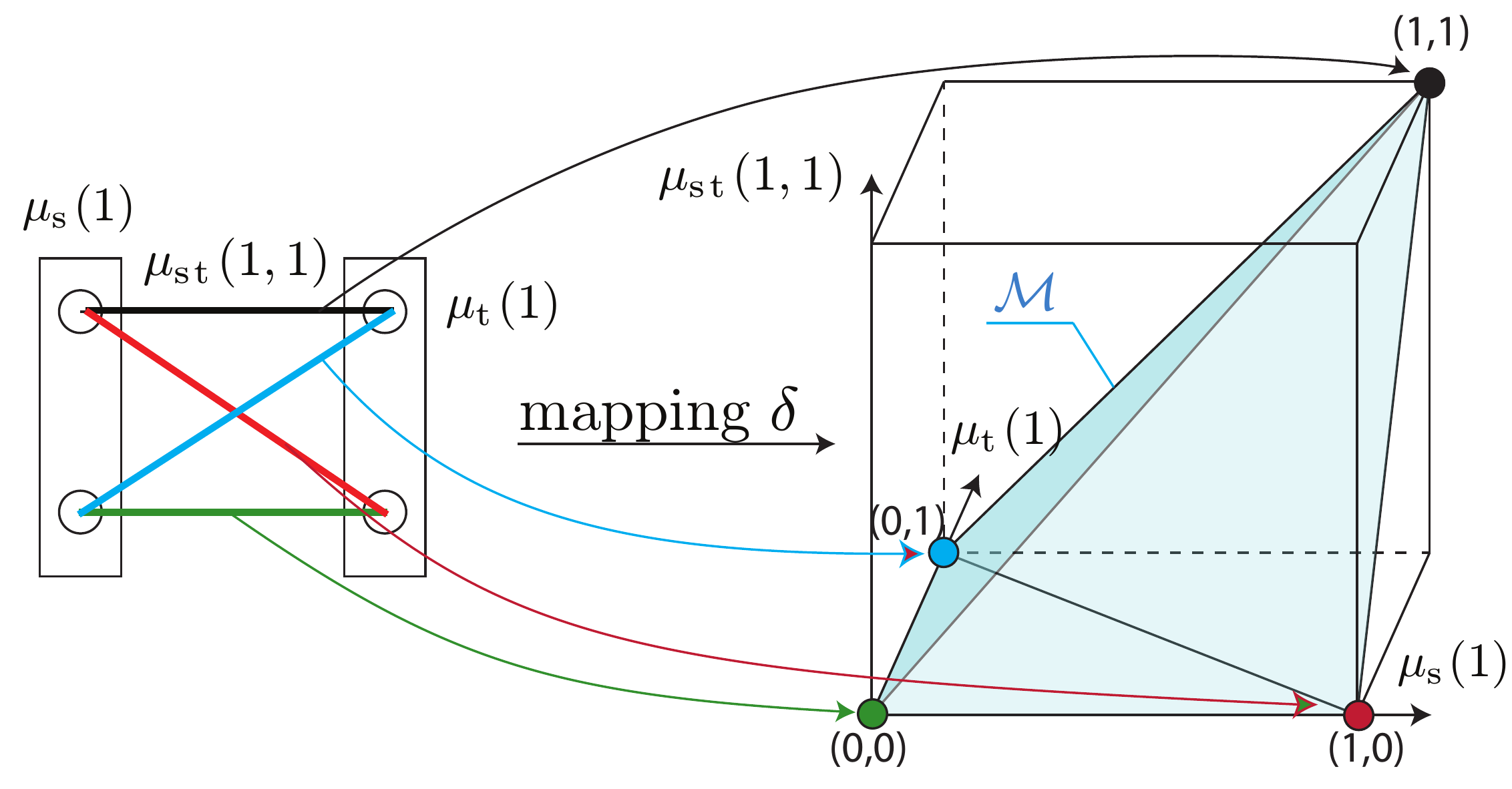}
\caption{Mapping $\delta$ embeds discrete labelings as points in the space $\Real^\I$. Projection onto components $\mu_s(1), \mu_t(1), \mu_{st}(1,1)$ is shown, the other components are dependent.} 
\label{f:embed1}
\end{figure}
\paragraph{LP Relaxation}
Let $\phi(x)\in\Real^\I$ be the vector with components $\phi(x)_0=1$, $\phi(x)_{s}\ind{i} = \leftbb x_s{=}i \rightbb$ and $\phi(x)_{st}\ind{\ij} = \leftbb (x_{s},x_t){=}ij \rightbb$, where $\leftbb\rightbb$ is the Iverson bracket.
Let $\<\cdot,\cdot\>$ denote the scalar product in $\Real^\I$. 
We can write the energy as
\begin{equation}
\f(x) = \<f,\phi(x)\>.
\end{equation}
The energy minimization can be expressed and relaxed as
\begin{equation}
\label{energy-linear}
\min_{x\in \LL} \<f,\phi(x)\>
 = \min_{\mu\in \delta(\LL) }\<f,\mu\>  = \min_{\mu\in \M }\<f,\mu\> \geq \min_{\mu\in\Lambda}\<f,\mu\>,
\end{equation}
where $\M = \conv \delta(\L)$ and $\Lambda$ is the {\em local} polytope that makes an outer approximation of $\M$.
We consider the standard Schlesinger's LP relaxation~\cite{Schlesinger-76}, where the polytope $\Lambda$ is given by the primal constraints in the following primal-dual pair: 
%
\begin{equation}\label{LP-big}
\begin{array}{lr}
\mbox{(LP-primal)} & \mbox{(LP-dual)} \\
\min \<f,\mu\>\tab\tab\tab\tab\tab = & \max \psi\\
\sum_{j}\mu_{st}(i,j) - \mu_s(i) = 0,  & \varphi_{st}(i) \in \Real,\\
\sum_{i}\mu_{st}(i,j) - \mu_t(j) = 0,  & \varphi_{ts}(j) \in \Real,\\
\sum_{i}\mu_{s}(i) -\mu_0 = 0, & \varphi_{s} \in \Real,\\
\mu_0 = 1, & \psi \in \Real,\\
\mu_s(i) \geq  0, & \hskip-1cm f_{s}(i) +\sum_{t\in \N(s)}\varphi_{st}(i)-\varphi_s \geq 0, \\
\mu_{st}(i,j) \geq  0, & \hskip-1cm f_{st}(i,j) -\varphi_{st}(i)-\varphi_{ts}(j) \geq 0, \\
\mu_0 \geq 0; & f_0 + \sum_{s}\varphi_s -\psi \geq 0.
\end{array}
\notag
\end{equation}
%
%
%
This relaxation is illustrated in \Figure{f:embed1}. 
We write it compactly as
\begin{equation}\label{LP}
\begin{array}{rrlr}
& \min \<f,\mu\> & = \hskip-0.3cm & \max \psi \,,\\
&
\setlength{\arraycolsep}{0.2em}
\begin{array}{rl}
A \mu &= 0\\
\mu_0 &= 1\\
\mu &\geq 0
\end{array}
& &
\setlength{\arraycolsep}{0.2em}
\begin{array}{rl}
\varphi & \in \Real^m\\
\psi &\in \Real \\
f - A\T \varphi - e_0 \psi & \geq 0\\
\end{array}
\end{array}
\tag{LP}
\end{equation}
where $A$ is $m \times |I|$ and $e_0\in\Real^\I$ is the basis vector for component $0$. 
Vector $f^{\varphi} := f - A\T \varphi$ is called an {\em equivalent transformation} ({\em reparametrization}) of $f$. There holds $\<f^{\varphi},\mu\> = \<f,\mu\> - \<\varphi,A\mu\> = \<f,\mu\>$ for all $\mu \in \Lambda$.
Because $\Lambda\supset \delta(\LL)$, it follows that $\f(x) = E_{f^{\varphi}}(x)$ for all $x\in\LL$. 
If there exists $\varphi$ such that $g = f^\varphi$ we write $g\equiv f$. In this case vectors $f$ and $g$ are different but they define equal energy functions $\f = \g$. 
See, \eg,~\cite{Werner-PAMI07} for more detail.
\par
Let $(\mu,(\varphi,\psi))$ be a feasible primal-dual pair. 
Complementary slackness for~\eqref{LP} states that $\mu$ is optimal to the primal and $(\varphi,\psi)$ to the dual iff 
\begin{subequations}\label{slackness}
\begin{align}\label{slack-a}
\mu_{s}(i) > 0\ \Rightarrow\ \ & f^\varphi_{s}(i) = 0,\\ 
\label{slack-b}
\mu_{st}\ind{\ij} > 0\ \Rightarrow\ \ & f^\varphi_{st}\ind{\ij} = 0,\\
\mu_0 >0 \ \Rightarrow\ \ & \psi = f_0 + \sum_{s}\varphi_s.
\label{slack-c}
\end{align}
\end{subequations}
Because a feasible dual solution satisfies $(\forall i')\ f^\varphi_{s}(i') \geq 0$, condition on the RHS\footnote{RHS = Right-hand side of an equation.} of~\eqref{slack-a} imply that label $i$ is {\em minimal} for $f^\varphi$. Similarly, in case of~\eqref{slack-b} we say that $ij$ is a minimal pair.
Implication~\eqref{slack-c} has its premise always satisfied.
\par

%% file: tex/improving_mappings.tex
\section{Improving Mapping}
\begin{definition}\label{def:improving mapping}
A mapping $p \colon \LL \to \LL$ is called {\em (weakly) improving} for $f$ if\\[-5pt]
\begin{equation}\label{improving mapping}
(\forall x\in\LL) \tab \f(p (x)) \leq \f (x),
\end{equation}
and {\em strictly improving} if
\begin{equation}\label{s-improving mapping}
(p(x)\neq x) \ \Rightarrow\ \f(p (x)) < \f (x),
\end{equation}
\end{definition}\noindent
%
We will consider {\em pixel-wise} mappings, of the form
$p(x)_s = p_s(x_s)$,
where $(\forall s\in \V)$ $p_s \colon \L_s \to \L_s$. Furthermore, we restrict to idempotent mappings, \ie, satisfying $p\circ p = p$, where $\circ$ denotes composition.

\begin{statement}Let $p$ be an improving pixel-wise idempotent mapping. Then there exists an optimal solution $x^*$ such that
\begin{equation}\label{improving-opt}
(\forall i) \tab  p_s(i) \neq i \ \Rightarrow\ \  x^*_s \neq i.
\end{equation}
\end{statement}
In case $p$ is strictly improving any optimal solution $x^*$ satisfies~\eqref{improving-opt}.
\begin{proof}
Let $x$ be optimal. Then $x^* = p(x)$ is optimal as well. By idempotency, $x^*$ satisfies $p(x^*) = x^*$. Condition~\eqref{improving-opt} is equivalent to $(\forall i)$ $x^*_s=i$\  $\Rightarrow$\  $p_s(i)=i$.
If $p$ is strictly improving, for any optimal solution $x^*$ there must hold $p(x^*) = x^*$, otherwise $\f(p(x^*)) < \f(x^*)$. 
\qed
\end{proof}
It follows that knowing an improving mapping, we can eliminate labels $(s,i)$ for which $p_s(i)\neq i$ as non-optimal.
Given a mapping $p$, the verification of the improving property is NP-hard: in case of binary variables it includes NP-hard decision problem of whether a partial assignment is an autarky~\cite{Boros:TR06-probe}. 
We need a simpler sufficient condition. 
It will be constructed by embedding the mapping into the linear space and applying a relaxation there.
\subsection{Relaxed Improving Mapping}
\begin{definition}A {\em linear extension} of $p\colon \LL \to \LL$ is a linear mapping $P\colon \Real^\I \to \Real^\I$ that satisfies
\begin{equation}\label{p-ext}
(\forall x\in\LL)\tab \delta(p(x)) = P\delta(x).
\end{equation} 
\end{definition}
\begin{figure}
\includegraphics[width=\linewidth]{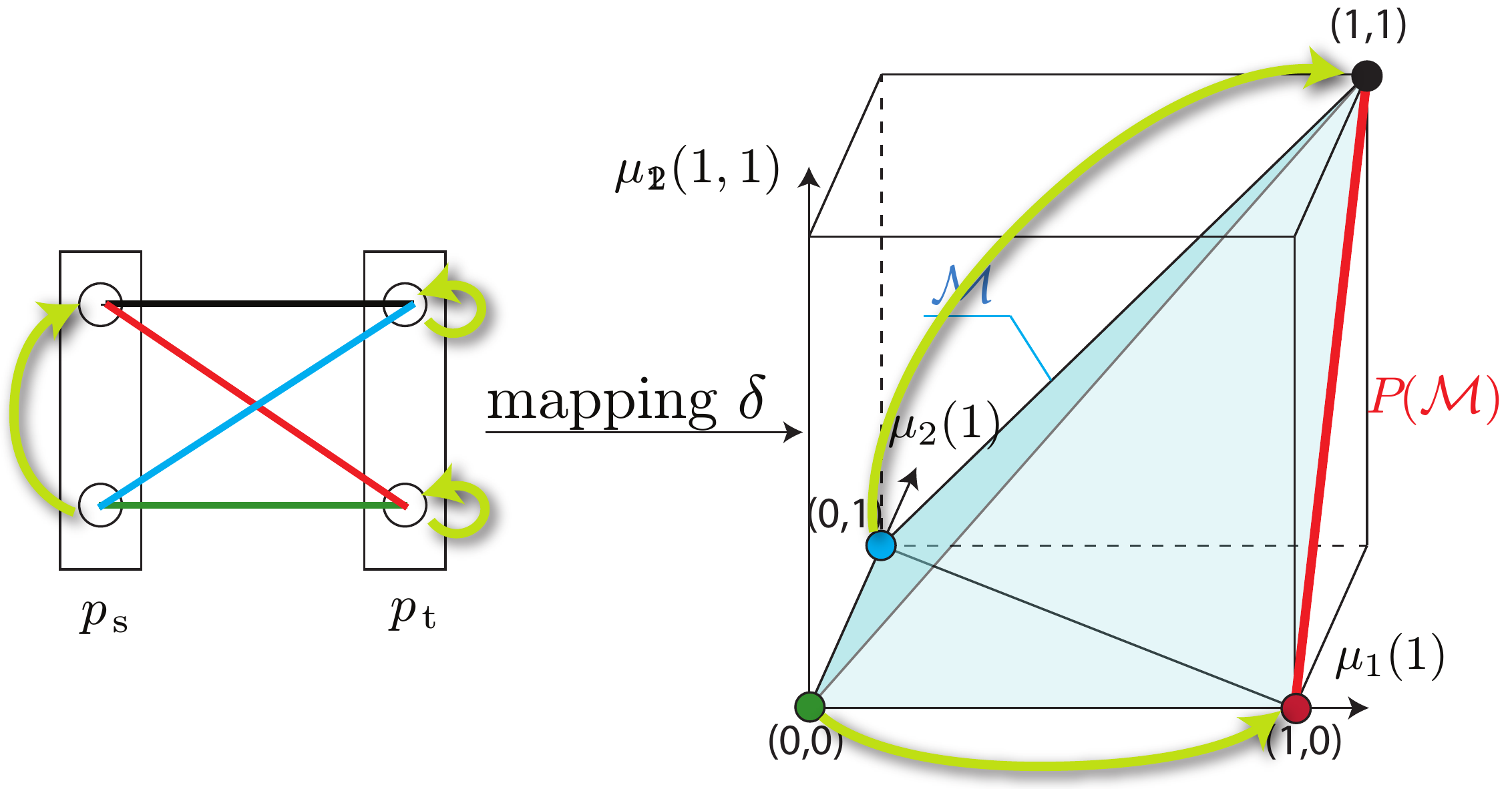}
\caption{Discrete map $p$ sends some labelings to other (the green labeling to red and the blue one to black). 
There is a corresponding linear map $P\colon \Real^\I\to\Real^\I$ (unique on $\aff(\M)$) wit this action -- the oblique projection onto the red facet.}\label{f:embed2}
\end{figure}
See \Figure{f:embed2} for illustration. We will only use the following particular linear extension for a pixel-wise mapping $p\colon \L \to \L$, which will be denoted $[p]$.
For each $p_s$ define matrix $P_s \in \Real^{\L_{s} \times \L_s}$ as $P_{s,ii'} = \leftbb p_s(i') = i\rightbb$.
The linear extension $P=[p]$ is given by
\begin{equation}\label{pixel-ext}
\begin{aligned}
& (P\mu)_{0} = \mu_0,\\
& (P\mu)_{s}(i) = P_{s} \mu_{s},\\
& (P\mu)_{st}(ij) = P_{s} \mu_{st} P_{t}\T.
\end{aligned}
\end{equation}
%
Linear maps of the form~\eqref{pixel-ext} with general matrices $P_s$  satisfying $P_s \geq 0$ and $1\T P_s = 1$ will be called {\em pixel-wise}. To verify that~\eqref{p-ext} holds true we expand the components as follows. $(P\delta(x))_{s}(i) = \sum_{i'\in \L_s} \leftbb p_s(i'){=}i \rightbb \leftbb x_s{=}i' \rightbb = \leftbb p_s(x_s){=}i \rightbb = \delta(p(x))_s(i)$. Similarly, for pairwise components,
$(P\delta(x))_{st}(i,j) = 
\leftbb p_s(x_s){=}i \rightbb \leftbb p_t(x_t){=}j \rightbb = \delta(p(x))_{st}(i,j)$.
\par
Using the linear extension $P$ of $p$ we can write
\begin{equation}
\f(p(x)) = \<f, \delta (p(x))\>  = \<f, P \delta(x) \>. 
\end{equation}
This allows to express condition~\eqref{improving mapping} as 
\begin{equation}\label{improving-conj}
(\forall x\in\LL)\tab  \<f, P \delta(x) \> \leq  \<f, \delta(x) \>.
\end{equation}
We introduce a sufficient condition by requiring that this inequality is satisfied over a larger subset $\Lambda$.
\begin{definition}
A linear mapping $P \colon \Real^\I \to \Real^\I$ is a (weak) {\em $\Lambda$-improving} mapping for $f$ if
\begin{align}\label{L-improving-conj-a}
& (\forall \mu\in\Lambda)\tab  \<f, P \mu \> \leq \<f, \mu \>;
\end{align}
and is a {\em strict $\Lambda$-improving} mapping for $f$ if
\begin{align}\label{L-improving-conj-b}
&(\forall \mu\in\Lambda,\ P\mu \neq \mu) \tab  \<f, P \mu \> < \<f, \mu \>.
\end{align}
\end{definition}
The set of mappings for which~\eqref{L-improving-conj-a} (resp. ~\eqref{L-improving-conj-b}) is satisfied will be denoted \WI (resp. \SI). For convenience, we will use the term {\em relaxed improving} map, meaning it \wrt polytope $\Lambda$.
%
Note, this definition and some theorems are given for arbitrary linear maps, at the same time for the purpose of this paper it would be sufficient to assume pixel-wise maps of the form~\eqref{pixel-ext}.
Clearly,~\eqref{L-improving-conj-a} implies~\eqref{improving-conj} because $\delta(\LL) \subset \Lambda$ and for the linear extension $[p]$ it implies that $p$ is improving. 
Sets \WI and \SI are convex as they are intersections of half-spaces (respectively, closed and open). {\em Verification} of~\eqref{L-improving-conj-a} for a given $P$ 
can be performed via solving 
\begin{equation}\label{L-LP}
\min_{\mu\in\Lambda} \<(I-P\T)f,\mu \>
\end{equation}
and checking that the result is non-negative, \ie can be decided in polynomial time.

%% file: tex/unification.tex
\section{Special Cases}\label{S:unification}
In order to show that other methods in the literature are special cases of condition~\eqref{L-improving-conj-a} we first identify a pixel-wise idempotent mapping $p$ they construct.
Then we apply the following {\em trivial} sufficient condition for $[p]\in\WI$. 
%
\begin{statement}\label{S:sufficient}Let $(\forall u\in \V)$ $(\forall st\in \E)$
\begin{subequations}\label{eq:proj-sufficient}
\begin{align}\label{eq:sufficient_1}
(\forall i\in \L_u) \tab & f_{u}\ind{p_u(i)} \leq f_{u}\ind{i},\\
\label{eq:sufficient_2}
(\forall ij\in \L_{st})\tab & f_{st}\ind{p_s(i),p_t(j)} \leq f_{st}\ind{\ij}.
\end{align}
\end{subequations}
Then $[p] \in \WI$. If additionally
$
(p_u(i){\neq}i) \Rightarrow\ f_{u}\ind{p_u(i)} < f_{u}\ind{i}$ for all $s\in\V$, $i\in\L_s$,
then $[p] \in \SI$.
\end{statement}
\begin{proof}
Let $\mu \in \Lambda$. By multiplying~\eqref{eq:sufficient_1} by $\mu_{u}\ind{i}$ and summing over $u$ and $i$ and multiplying~\eqref{eq:sufficient_2} by $\mu_{st}\ind{\ij}$ and summing over $st$ and $ij$ we get
\begin{equation}\label{sufficient-result}
\<[p]\T f, \mu \> \leq \<f,\mu\>\,,
\end{equation}
which is equivalent to~\eqref{L-improving-conj-a}. Suppose $[p]\mu\neq \mu$. Then $\exists s\in \V, \exists i\in\L_s$ such that $\mu_{s}(i)>0$ and $p_s(i)\neq i$. Therefore there will be at least one strict inequality in the sum (with positive coefficient) and~\eqref{sufficient-result} will hold strictly.
\qed
\end{proof}\noindent
How this component-wise condition can be used to explain global methods? The trick is use it in combination with equivalent transformations. It turns out that this combination is very powerful and in fact characterizes $\WI$ (see \Section{S:char-s}).
%
We will consider mainly {\em weak} variants of all methods. We start with a simple local method.
\subsection{DEE}
There is a number of local sufficient conditions proposed that are generally referred to as {\em dead end elimination} (DEE)~\cite{Desmet-92-dee,Goldstein-94-dee,Lasters-95-dee,Georgiev-06-dee,Pierce-2000-dee}. We will consider Goldstein's {\em simple} DEE~\cite{Goldstein-94-dee}. This method for every pixel $s$ and labels $\alpha,\beta\in\L_s$ verifies the condition $(\forall x\in \L_{\N(s)})$
\begin{equation}\label{DEE-simple}
\begin{split}
f_s(\alpha)-f_s(\beta)
+\sum_{t\in\N(s)} [f_{st}(\alpha,x_t)-f_{st}(\beta,x_t)] \geq 0.
\end{split}
\end{equation}
If the condition is satisfied it means that a (weakly) improving switch from $\alpha$ to $\beta$ exists for an arbitrary labelling $x$. In that case, $(s,\alpha)$ can be eliminated, preserving at least one optimal assignment. 
\par
Let $p_s(\alpha) = \beta$, $p_s(i) = i$ for $i\neq \alpha$; and $p_t(i)=i$ for all $t\neq s$. Let $P=[p]$. We claim $P\in\WI$.
\begin{proof}
The condition~\eqref{DEE-simple} can be written as
\begin{equation}
\min_{x\in \LL} [f(x) - f(p(x))] \geq 0.
\end{equation}
\end{proof}
This minimization problem efficiently has a star structure (non-zero unary terms only for $s$ and pairwise terms for neighbors of $s$). It is equivalent therefore to
\begin{equation}
\min_{\mu \in \Lambda} \<f,\mu - P \mu \> \geq 0.
\end{equation}
\qed
\par
Similarly, the strict inequality~\eqref{DEE-simple} implies $P \in \SI$.
%
%
%
\subsection{QPBO Weak Persistency}
The {\em weak persistency} theorem~\cite{Nemhauser-75,Hammer-84-roof-duality} states the following.
Let $\L_s = \{0,1\} = \Bool$. Let $\mu\in \argmin_{\mu\in\Lambda}\<f,\mu\>$.
Let $O_s = \{i\in\Bool \mid \mu_s(i) > 0\}$.
Then
\begin{equation}
(\exists x\in \argmin_x \f(x))\ (\forall s\in\V) \ x_s\in O_s.
\end{equation}
In the case $|O_s|=1$ vector $\mu_{s}$ is necessarily integer and the theorem states that for such integer pixels, $x_s$ can be fixed accordingly.
\par
We define $p_s(i) = 0$ if $O_s=\{0\}$, $p_s(i) = 1$ if $O_s=\{1\}$ and $p_s(i) = i$ otherwise. We claim $[p] \in \WI$.
\begin{proof} Let $\varphi$ be a solution to LP-dual. 
We will show that the following component-wise inequalities hold:
\begin{subequations}
\label{qpbo-comp}
\begin{align}
\label{qpbo-comp-a}
f^\varphi_s (i) & \geq f^\varphi_s (p_s(i)),\\
\label{qpbo-comp-b}
f^\varphi_{st} (i,j) & \geq f^\varphi_{st} (p_{s}(i),p_t(j)).
\end{align}
\end{subequations}
Unary inequalities~\eqref{qpbo-comp-a} hold by construction of $p$ and complementary slackness~\eqref{slack-a}. Let us show pairwise inequalities~\eqref{qpbo-comp-b}. 
Let $y_{st} = p(x)_{st}$. Consider the following cases:
\begin{itemize}
\item $|O_s| = 1$, $|O_t| = 1$. Necessarily, $\mu_s(y_s) = 1$ and $\mu_t(y_t) = 1$. By feasibility, $\mu_{st}(y_{st}) = 1$. By complementary slackness, $f^{\varphi}_{st}(y_{st}) \leq f^{\varphi}_{st}(x_{st})$.
\item $|O_s| = 1$, $|O_t| = 2$. Necessarily, $\mu_s(y_s) = 1$. By feasibility, $\mu_{st}(y_s,i) > 0$ for all $i\in\Bool$. By complementary slackness, $f^{\varphi}_{st}(y_s,x_t) \leq f^{\varphi}_{st}(x_s,x_t)$.
\item $|O_s| = 2$, $|O_t| = 2$. In this case $y_{st}=x_{st}$ and hence $f^{\varphi}_{st}(y_{st}) = f^{\varphi}_{st}(x_{st})$.
\end{itemize}
Therefore, for every $st\in\E$ we have $f^{\varphi}_{st}(y_{st}) \leq f^{\varphi}_{st}(x_{st})$. 
By \Statement{S:sufficient},~\eqref{qpbo-comp} implies $p \in \WI$.
\qed
\end{proof}
%
\subsection{QPBO Strong Persistency}
%
Let $(\mu,\varphi)$ be a feasible primal-dual pair for~\eqref{LP}. This pair is called {\em strictly complementary} if
\begin{subequations}\label{sslackness}
\begin{align}\label{sslack-a}
\mu_{s}(i) > 0\ & \Leftrightarrow\ f^\varphi_{s}(i) = 0, \\
\label{sslack-b}
\mu_{st}\ind{\ij} > 0\ & \Leftrightarrow\ f^\varphi_{st}\ind{\ij} = 0\\
 & \psi = f_0 + \sum_{s}\varphi_s.
\end{align}
\end{subequations}
Clearly, strictly complementary pair is complementary and thus it is optimal. Such pair always exists and can be found by interior point algorithms. It is known that $\mu$ is a relative interior point of the primal optimal facet and $\varphi$ is relative interior point of the dual optimal facet. 
\par
The {\em strong persistency} theorem~\cite{Nemhauser-75,Hammer-84-roof-duality} considers pixels $s\in\V$ such that $\mu'_{s}$ is integer in {\em all} solutions $\mu'$ to the LP-relaxation. It is seen that $\mu'_{s}(i)>0$ in some optimal solution $\mu'$ iff $\mu_{s}(i)>0$ for the relative interior optimal solution $\mu$. Clearly, the solution $\mu$ has the minimum number of integer components of all solutions. Let
\begin{equation}\label{s-O1}
O_s = \{i\in\L_s \mid \mu_s(i) > 0\}.
\end{equation}
For a strictly complementary pair,~\eqref{s-O1} defines the same sets as
\begin{equation}\label{s-O2}
O_s = \argmin_{i}f^{\varphi}_{s}(i).
\end{equation}
So we need either primal or dual relative interior optimal point. 
The theorem can be formulated as follows. 
%
%
Let $(\mu,\varphi)$ be a strictly complementary primal-dual pair. 
Let $O_s$ be defied by~\eqref{s-O2}. Then 
\begin{equation}
(\forall x\in \argmin_x \f(x))\ (\forall s\in\V) \ x_s\in O_s.
\end{equation}
%
\par
Let us consider the pixel-wise mapping $p$:
\begin{equation}
p_s(i) = \begin{cases}
0, & O_s=\{0\},\\
1, & O_s=\{1\},\\
i, & O_s=\{0,1\}.
\end{cases}
\end{equation}
\begin{statement} We claim that $[p] \in \SI$.
\end{statement}
\begin{proof}
By construction of $p$ and $O_s$ we have $(\forall i)$
\begin{align}
\label{qpbo-comp-a-s}
f^\varphi_s (i) & \geq f^\varphi_s (p_s(i)).
\end{align}
If $p_s(i)\neq i$, then $i\notin\ O_s$ and inequality~\eqref{qpbo-comp-a-s} is strict.
We also have the pairwise inequalities~\eqref{qpbo-comp-b} implied by non-strict complementary slackness as in the weak persistency case. 
%
%
By \Statement{S:sufficient}, it follows that $p \in \SI$.
\qed
\end{proof}
We can verify that mapping $p$ is the maximum because any mapping that is larger violates necessary conditions of \SI to be given in \Lemma{necessary-LI}. 
Therefore it is the solution to {\sc max-si}.
\par
%
\subsection{MQPBO}
MQPBO method~\cite{kohli:icml08} extends partial optimality properties of QPBO to multi-label problems via the reduction of the problem to $0$-$1$ variables. The reduction is for a predefined ordering of labels. The method outputs two labelings $x^{\rm min}$ and $x^{\rm max}$ with the guarantee that there exists optimal labeling $x$ that satisfy $x_s \in [x^{\rm min}_s,\, x^{\rm max}_s]$. The improving mapping the method constructs has the form $p(x) = (x \vee x^{\rm min}) \wedge x^{\rm max}$. Because the reduction is component-wise and we showed component-wise inequalities~\eqref{qpbo-comp} for QPBO, it can be shown that component-wise conditions hold for $p$ and therefore $[p]\in \WI$. 
\par
Let $f$ be a multi-label problem and $g$ the equivalent binary (having $\{0,1\}$-valued decision variables) energy minimization problem as defined in~\cite{kohli:icml08}.
The mapping of multi-valued to binary labelings is given by $z_{s,i}(x) = \leftbb x_s {>} i \rightbb$. The corresponding mapping of multi-label relaxed labelings $\mu$ to relaxed labelings $\nu$ of the binary problem is given~\citesuppl{Shekhovtsov-07-binary-TR} as follows. 
For index $i\in\L_s = \{0,1,\ldots,K-1\}$ introduce the following sets of labels:
\begin{subequations}
\begin{align}
&L_{s}(i,0)  = \{0,\ldots,i\},\\
&L_{s}(i,1)  = \{i+1,\ldots,K-1\}.
\end{align}
\end{subequations}
\noindent The vector $\nu = \Pi\mu$ is defined as
\begin{subequations}\label{mapping Pi}
\begin{align}
\nu_{(s,i)}(\alpha)  =\!\!\sum\limits_{i'\in L_{s}(i,\alpha)} \mu_{s}(i'),\\
\nu_{(s,i)(t,j)}(\alpha,\beta)=\!\!\sum\limits_{\begin{subarray}{c}i'\in L_{s}(i,\alpha) \\ j'\in L_{t}(j,\beta)\end{subarray}} \mu_{st}(i',j'),
\end{align}
\end{subequations}
where $i$ and $j$ range in $\tilde\L_s = \tilde\L_t = \{0,1,\dots, K-2\}$. The mapping $\Pi$ is consistent with the mapping $z_{s,i}(x) = \leftbb x_s {>} i \rightbb$ in the sense that $\Pi \phi(x) = \phi(z(x))$ for all $x$.
Using the mapping $\Pi$, the equivalence of multi-label and binary problems $\forall x\in\LL$ $\f(x) = \g(z(x))$ is expressed as
\begin{equation}\label{K2-linearized}
(\forall \mu \in \Lambda) \tab \<f,\mu\> = \<g,\Pi \mu \>\,.
\end{equation}
Let QPBO method for $g$ construct labelings $z^{\rm min}_s$, $z^{\rm max}_s$ such that the mapping $q \colon z \to (z \vee z^{\rm min})\wedge z^{\rm max}$ is strictly (resp. weakly) improving for $g$. 
Let 
\begin{equation}
x^{\rm min}_{s} = \sum_{i'\in \tilde \L}z^{\rm min}_{s,i'}, \tab\tab
x^{\rm max}_{s} = \sum_{i'\in \tilde \L}z^{\rm max}_{s,i'}.
\end{equation} 
It was shown~\cite{kohli:icml08} that the mapping $p\colon x \to (x \vee x^{\rm min}) \wedge x^{\rm max}$ is strictly (resp. weakly) improving for $f$.
\begin{statement}
 We claim that the linear extension $[p]$ is in $\SI$ (resp. in $\WI$).
%
%
%
\end{statement}\noindent
\begin{proof}
%
Let $\tilde g$ be the problem equivalent to $g$ for which component-wise inequalities~\eqref{eq:proj-sufficient} hold (as proven to exist for QPBO). By equivalence~\eqref{K2-linearized}, we have
\begin{equation}\label{mqpbo-eq1}
(\forall \mu \in \Lambda) \tab \<f,\mu\> = \<g, \Pi\mu \> = \<\tilde g, \Pi \mu \>.
\end{equation}
Let $\tilde f = \Pi\T \tilde g$. From~\eqref{mqpbo-eq1} we conclude that $\tilde f \equiv f$. We expand now components of $\tilde f$ using~\eqref{mapping Pi} and component-wise conditions for $\tilde g$: 
\begin{equation}\label{PiT f expand}
\begin{aligned}
\tilde f_{s}(x_s) &= \sum_{i'} \tilde g_{(s,i')}(z_{s,i'}) \\
& \geq  \sum_{i'} \tilde g_{(s,i')}((z_{s,i'}\vee z^{\rm min}_{s,i'})\wedge z^{\rm max}_{s,i'} )\\
&= \tilde f_{s}((x_s\vee x^{\rm min})_s\wedge x^{\rm max}_{s})\,.
\end{aligned}
\end{equation}
Similarly, for pairwise terms: $\tilde f_{st}(x_{st}) =$
\begin{align}
& \sum_{i'j'} \tilde g_{(s,i')(t,j')}\ind{z_{(s,i')(t,j')}} \\
\notag
& \geq  \sum_{i',j'} \tilde g_{(s,i')(t,j')}\Big((z_{(s,i')(t,j')}\vee z^{\rm min}_{(s,i')(t,j')})\wedge z^{\rm max}_{(s,i')(t,j')} \Big) \\
\notag
& = \tilde f_{st}((x_{st}\vee x^{\rm min})_{st}\wedge x^{\rm max}_{st}).
\end{align}
Therefore, component-wise inequalities hold for $\tilde f$. By the sufficient condition~\eqref{eq:proj-sufficient}, we conclude that
\begin{equation}
(\forall \mu \in \Lambda) \tab \<\tilde f, [p] \mu \> \leq \<\tilde f, \mu \>\,,
\end{equation}
therefore $[p]$ is in $\WI$. In the case that $q$ is strictly improving, from $p(x)\neq x$ follows $q(z)\neq z$ and one of the inequalities~\eqref{PiT f expand} holds strictly. In this case $[p]\in\SI$.
\qed
\end{proof}
%
%
%
%
%
\subsection{Auxiliary Submodular Problems} There were several methods proposed~\cite{Kovtun03,Kovtun-10} that differ in detail. All methods construct an {\em auxiliary} submodular (in a given ordering of labels) energy $\g$. A minimizer $y$ of $\g$ has the property that $\g(x \vee y) \leq \g(x)$, implied by submodularity. It follows that mapping $p(x) = x \vee y$ is improving for $g$. 
The construction of the auxiliary function (to be specified) ensures that improvement in $f$ is at least as big as improvement in $g$ for the full family of mappings $x\mapsto x\vee y$, assuming $y$ is not known. 
It follows that $p$ is improving for $f$ and thus provides partial optimality. 
\par
Let $P=[p]$. We claim $P\in\WI$.
\begin{proof}
First, we show that the auxiliary property of $g$ implies
\begin{equation}\label{Kovtum-L-aux}
(\forall \mu \in \Lambda) \ \ \<f,P\mu - \mu\> \leq \<g,P\mu - \mu\>.
\end{equation}
The auxiliary function $\g$ in~\cite{Kovtun-10} satisfies the following component-wise inequalities:
\begin{subequations}\label{Kovtun-aux}
\begin{equation}
\begin{aligned}\label{Kovtun-aux-a}
&&& (\forall s \in \V) \ \ (\forall i \in \L_s) \ \ (\forall i'\in \K_s)\\
&&& \ \ (f-g)_{s}(i \vee i') \leq (f-g)_s(i),
\end{aligned}
\end{equation}
\begin{equation}
\begin{aligned}\label{Kovtun-aux-b}
&&& (\forall st \in \E), (\forall ij\in \L_{st})\ \ (\forall i'j'\in \K_{s}\times{\K_t})\\
&&& \ \ (f-g)_{st}(i\vee i', j\vee j') \leq (f-g)_{st}\ind{\ij},
\end{aligned}
\end{equation}
\end{subequations}
where $\K_s \subset \L_{s}$ depend on a particular method. All methods ensure that $y_s\in \K_s$.
Let $\mu\in\Lambda$. By multiplying inequalities~\eqref{Kovtun-aux-a} for $i'=y_s$ with $\mu_{s}(i)$, inequalities~\eqref{Kovtun-aux-b} for $i'j'=y_{st}$ with $\mu_{st}(ij)$ and adding we obtain~\eqref{Kovtum-L-aux}.
\par
Second, we show that $P\in\WIg$. Recall that $\g$ is submodular and $y$ is a minimizer. LP-relaxation for $g$ is tight, therefore there exists dual $\varphi$ such that $(\delta(y),\varphi)$ satisfies complementary slackness~\eqref{slackness}. 
Let $\tilde g = g^{\varphi}$. 
By~\eqref{slack-a}, $\tilde g_{s}(x_{s} \vee y_{s}) \leq \tilde g_{s}(x_{s})$.
For the pairwise components we inspect the four cases in order to prove $\tilde g_{st}(x_{st} \vee y_{st}) \leq \tilde g_{st}(x_{st})$:
\begin{itemize}
\item $x_{st} \geq y_{st}$: in this case $x_{st} \vee y_{st} = x_{st}$.
\item $x_{st} < y_{st}$: in this case $x_{st} \vee y_{st} = y_{st}$, which is minimal. 
\item $x_{s} < y_{s}$, $x_{t} \geq y_{t}$: in this case $x_{st} \vee y_{st} = (y_{s}, x_{t})$. The submodularity inequality 
$\tilde g_{st}(x_{st})+\tilde g_{st}(y_{st}) \geq \tilde g_{st}(x_s, y_t)+\tilde g_{st}(y_s, x_t)$ and 
minimality of $y_{st}$ imply $\tilde g_{st}(x_{st}) \geq \tilde g_{st}(y_s, x_t)$.
\item $x_{s} \geq y_{s}$, $x_{t} < y_{t}$: similar to the above.
\end{itemize}
By \Statement{S:sufficient}, $p \in \WIx{\tilde g} = \WIg$. From~\eqref{Kovtum-L-aux} follows $p \in \WI$.
\qed
\end{proof}

The {\em one-against-all-binary} method~\cite{Kovtun-10} restricts $y_s$ to $\{0,\bar y_s\}$ for some fixed labeling $\bar y$ (\eg. $\bar y_s = \alpha$ for all $s$), let us call it the {\em test labeling}. The labels are then reordered such that $\bar y_s$ becomes the highest label and sets $K_s$ are chosen to be $\{0,K-1\}$.
Furthermore, $g$ is additionally constrained to be equivalent to a problem with two labels. In this case the result of the method depends only on the choice of $\bar y$ and not on the actual ordering.

%
%
\subsection{Iterative Pruning}
Iterative Pruning method~\cite{Swoboda-13} was originally proposed for the Potts model: $f_{st}(i,j) = \gamma_{st} \leftbb i{\neq}j \rightbb$. It constructs a subset $\A \subset \V$, a labeling $y$ on $\A$ and an auxiliary energy $\g$:
\begin{align}
(\forall s\in \A) \tab & g_{s} = f_{s},\\
\notag
(\forall st\in \E,\ s\in \A, t\in \A) \tab & g_{st} = f_{st},\\
\notag
(\forall st\in \E,\ s\in \A, t \notin \A,\,\forall ij) \tab & g_{st}(i,j) = \gamma_{st} \leftbb i = y_s \rightbb,
\end{align}
with remaining terms set to zero. It can be seen that energy $\g$ depends on the assignment of $y$ only on the {\em boundary} $\partial\A = \{s\in \A \mid \exists st\in \bar\E, t\notin \A \}$.
Let us extend $y$ to $\V$ in an arbitrary way, \eg, by $y_{\V\backslash \A} = 0$. 
The sufficient condition of~\cite{Swoboda-13} imply that $\delta(y) \in \argmin_{\mu\in\Lambda}\<g,\mu\>$ (the relaxation is tight). 
We construct mapping $p$ as
\begin{equation}
p_s(i) = \begin{cases}
y_{s} & \IF s \in \A,\\
i & \IF s \notin \A,
\end{cases}
\end{equation}
\ie, $p$ replaces part of labeling $x$ on $\A$ with the labeling $y$. Let $P=[p]$. We claim that $P\in\WI$.
\begin{proof}
We first show that $g$ is auxiliary for $f$ in the same sense as for the method~\cite{Kovtun-10}. We trivially have $f_s(p_s(i))-f_s(i) = g_s(p_s(i))-g_s(i)$. We also have equality of pairwise terms $f_{st}(p(x)_{st})-f_{st}(x_{st}) = g_{st}(p(x)_{st})-g_{st}(x_{st})$ for $st\in \E$ in all of the following cases:
(a) $s\in \A$ and $t \in \A$;
(b) 
$s\notin \A$ and $t \notin \A$;
(c) $s\in \A$ and $t \notin \A$, $x_s = y_s$.
It remains to verify the inequality for boundary pairs $s\in\A$, $t\notin \A$ in the case $x_s\neq y_s$. We have
\begin{equation}
\begin{aligned}
& f_{st}(x_{st}) - f_{st}(p(x)_{st})\\
& \geq \min_{ij\in \L_{st}}\big(f_{st}(i,j) - f_{st}(p_s(i),p_t(j)) \big) = -\gamma\\
& = g_{st}(x_{st}) - g_{st}(p(x)_{st}).
\end{aligned}
\end{equation}
It follows that~\eqref{Kovtum-L-aux} holds. 
The second step is to show that $P\in\WIg$. 
By assumption, we have $\delta(y) \in \argmin_{\mu\in\Lambda}\<g,\mu\>$. 
Given a labeling $x$, mapping $p$ replaces part over $\A$ to the optimal labeling $y$.
It follows that $(\forall \mu\in\Lambda)$\ \ $\<g,P\mu\> = \<g,\delta(y)\> \leq \<g,\mu\>$. 
Combined with~\eqref{Kovtum-L-aux}, we obtain $P\in\WI$.
\qed
\par\noindent
\end{proof}

%
%

%% file: tex/characterization.tex
\section{Characterization}\label{S:char-s}
We introduced component-wise sufficient conditions~\eqref{eq:proj-sufficient} and observed while considering different methods that it was often possible to find a reparametrization of the problem such that these conditions hold. This is not a coincidence.
\begin{theorem}
Let $P=[p]$, $p$ idempotent and $P \in \WI$. Then exists $\varphi$ such that
\begin{equation}
P\T f^{\varphi} \leq f^{\varphi}.
\end{equation}
\end{theorem}
\begin{proof} The proof uses a representation of the verification problem~\eqref{L-LP} introduced in Section{sec:properties}. Let $g = (I-P\T)f$. The steps of the proof are given by the following chain:
\begin{equation}\label{proj2-chain}
\begin{array}{ll}
& \min\limits_{\begin{subarray}{l} A \mu = 0 \\ \mu\geq 0\end{subarray}} \<f-P\T f,\mu \>
\stackrel{\rm (b)}{=} 
\min\limits_{\begin{subarray}{l} A P \mu =0 \\ A(I-P)\mu = 0\\ \mu\geq 0 \\ \ \end{subarray}} \<f-P\T f,\mu \> \\
& \stackrel{\rm (c)}{=} 
\min\limits_{\begin{subarray}{l} \\ A(I-P)\mu = 0\\ \mu\geq 0\\ \ \end{subarray}} \<f-P\T f,\mu\>
\stackrel{\rm (d)}{=} \max\csub{\varphi \\ (I-P\T)(f-A\T \varphi) \geq 0 } 0\,.
\end{array}
\end{equation}
On the LHS we have problem~\eqref{L-improving-coni} which is bounded because $P\in\WI$. The value of the problem in this case equals zero. Under conditions of the theorem, equalities (b), (c) essentially claims boundedness of the other two minimization problems in the chain. Equality~(b) is verified as follows. Inequality $\leq$ holds because $A\mu - A P \mu = 0$ and $A P \mu = 0$ implies $A\mu = 0$. On the other hand, $P$ preserves all constraints of $\Lambda$ and therefore $A\mu=0$ $\Rightarrow$ $A P \mu = 0$.
\par
Equality~(c) is the key step. We removed one constraint, therefore $\geq$ trivially holds. Let us prove $\leq$. Let $\mu$ be feasible to RHS of equality (c). Let $\mu=\mu_1 + \mu_2$, where
\begin{equation}
\begin{aligned}
\mu_1 = P \mu\,;\ \ \ \ \mu_2 = (I-P) \mu\,.
\end{aligned}
\end{equation}
There holds
\begin{equation}
\begin{aligned}
(I-P)\mu_1 & = (I-P)P \mu  = (P-P^2) \mu = 0\,,\\
P \mu_2 & = P (I-P) \mu = 0\,,
\end{aligned}
\end{equation}
\ie, $\mu_1 \in \Null(I-P)$ and $\mu_2 \in \Null(P)$. Let us construct $\mu_1'$ as follows. Let $\gamma = $
\begin{align}
\notag
 &\max\big\{ \max_{st, ij} |\L_{st}|(\mu_{1})_{st}\ind{\ij}, \max_{s,i} |\L_s| (\mu_1)_s(i), (\mu_1)_0 \big\}\,,\\
& (\mu_1')_{st} = \gamma / |\L_{st}|\,,\\
\notag
& (\mu_1')_{s} = \gamma / |\L_{s}|\,,\\
\notag
& (\mu_1')_0 = \gamma\,.
\end{align}
By construction,
\begin{equation}
(\mu_1') \geq \mu_1 \tab\tab \AND \tab\tab A \mu_1' = 0\,.
\end{equation}
Let $\mu_1'' = P \mu_1'$. Because $P \geq 0$, we have
\begin{equation}
\mu_1'' = P \mu_1' \geq P \mu_1 = \mu_1\,.
\end{equation}
It also follows that $A P \mu_1'' = A P P \mu_1' = A P \mu_1' = 0$ and $(I-P) \mu_1'' = (I-P)P \mu_1' = 0$.
Let $\mu^* = \mu_{1}''+\mu_{2}$. It preserves the objective,
\begin{align}
\<f-P\T f,\mu^*\> &= \<f,(I-P)(\mu_{1}''+\mu_{2}) \>  \\
\notag
&= \<f,(I-P)\mu_{2}\> = \<f,(I-P)\mu\>\,.
\end{align}
We also have that 
\begin{equation}
\begin{aligned}
\mu^* = \mu_{1}''+\mu_{2} \geq \mu_1 + \mu_{2} = \mu \geq 0\,,\\
A(I-P)\mu^* = A(I-P)\mu_2 = A(I-P)\mu = 0\,,\\
A P\mu^* = A P \mu_1'' = 0\,.
\end{aligned}
\end{equation}
Therefore, $\mu^*$ satisfies all constraints of the LHS of equality (c).
\par
Equality~(d) is the duality relation that asserts that the maximization problem on the RHS is feasible,
which is the case \iff
\begin{equation}
(\exists g \equiv f) \tab (I-P\T) g \geq 0\,.
\end{equation}
\qed
\end{proof}

%% file: tex/properties.tex
\section{Maximum Improving Mapping}
Having a more powerful sufficient condition, which can be verified in polynomial time, how do we find a map that satisfies it? How do we find the map that delivers {\em the largest} partial optimal assignment, or, equivalently, eliminates the maximum number of labels as non-optimal? 
Recall that the label $(s,i)$ is eliminated by pixel-wise mapping $p$ if $\leftbb p_s(i){\neq}i \rightbb$. We therefore formulate the following {\em maximum persistency} problem:
%
\begin{equation}\label{best L-improving}
\tag{{\sc max-wi}}
\max_{p} \sum_{s,i} \leftbb p_s(i){\neq i}\rightbb \tab \mbox{s.t. $[p]\in\WI$.}\quad
\end{equation}
The strict variant, with constraint $[p]\in\SI$, will be denoted {\bf\sc max-si}.
The problem may look difficult, however, we will be able to solve it in polynomial time for some types of maps covering nearly all types that appeared in the previous section:
%
%
%
\begin{itemize}
\item {\em all-to-one maps}. Set $\P^{1,y}$ of maps of the form $p \colon x \mapsto x [\A \leftarrow y]$ for all $\A\subset\V$ and fixed $y\in \LL$. 
\item {\em subset-to-one maps}. 
Let $V = \{(s,i)\mid s\in\V, i\in \L_s\}$. Let $\xi \in\{0,1\}^V$. Mapping $p_\xi$ in every pixel either preserves the label or switches it to $y_s$:\\[-10pt]
\begin{equation}
p_\xi(x)_s = 
\begin{cases}
y_s &\ \IF\ \xi_{s x_s} = 1,\\
x_s &\ \OTHERWISE.
\end{cases}
\end{equation}
Vector $(\xi_{si}\mid i\in\L_s)$ serves as the indicator of a subset of labels in pixel $s$ that are mapped to $y_s$. The set $\P^{2,y}$ of all such maps is considered.
%
\item {\em all-to-one-unknown} maps. Set $\P^1 = \bigcup_{y\in\LL} \P^{1,y}$. 
\end{itemize}
Additionally, we define {\em subset-to-one-unknown} maps as the set $\P^2 = \bigcup_{y\in\LL} \P^{2,y}$. This set is considered merely to draw the boundary between solvable and unsolvable cases of maximum persistency problem. All complexity results are summarized in Table~\ref{table2}. 
%
%
%
%
%
\begin{table}
\begin{center}
\small
\setlength{\tabcolsep}{3pt}
\begin{tabular}{|c|c|c|}
\hline
problem type & {\sc max-si} & {\sc max-wi}\\
\hline
$K=2$ & P (QPBO)  & P (QPBO) \\
$K = 3$ & ? & NP-hard \\
$K>3$ & NP-hard & NP-hard \\
$\P^{1,y}$ & P ($\varepsilon$-L1) & P (L1) \\
$\P^{2,y}$ & P ($\varepsilon$-L1) & P (L1) \\
$\P^1$ & P (nec. cond. + $\varepsilon$-L1) & NP-hard \\
$\P^2$ & NP-hard & NP-hard \\
\hline
\end{tabular}
\end{center}\caption{Complexity of maximum persistency problem. Notation $K=2$ means the class of problems with 2 labels and arbitrary maps. In brackets we denote the respective polynomial method, see~\S\ref{sec:best mapping LP}.}\label{table2}
\end{table}
We see that as soon as $K>3$ the problem with unconstrained maps becomes intractable. 
We also see that the complexity jumps with the number of possible destinations for each label increased. Note, in case of all-to-one-unknown maps the difference between strict and weak conditions results in a different complexity class!
%
\section{Algorithms}\label{sec:best mapping LP}
\paragraph{Case $K=2$}
For the case of two labels ($K\,{=}\,2$), problem {\sc max-si} (resp. {\sc max-wi}) can be solved by finding solution to~\eqref{LP} with the minimum (resp. maximum) number of integer components. 
This corresponds to finding specific cuts in the network flow model~\cite{Boros:TR91-maxflow},~\cite[\S 2.3]{Kolmogorov-Rother-07-QBPO-pami}. Finding the relaxed solution with the maximum number of integer components was proven polynomial by Picard and Queyranne~\cite{Picard-77} in the context of vertex packing problem. 
We extend this proof to general quadratic pseudo-Boolean functions. 
\begin{statement}Let $\mu^1,\mu^2$ be two solutions to~\eqref{LP}. Let us denote sets where these solutions are integral as $U = \{s\in\V\mid \mu^1_s(i)\in\Bool \}$ and $V = \{s\in\V\mid \mu^2_s(i)\in\Bool \}$. Let $x^1$ and $x^2$ be corresponding partial labelings. Then there exists a solution $\mu$ such that its integral part $\A = \{s\in\V\mid \mu_s(i)\in\Bool \}$ is the union $U \cup V$.
\end{statement}
\begin{proof}
We construct $\mu$ as follows 
\begin{align}
& \mu_s = \begin{cases}
\mu^1_s, & s\in U,\\
\mu^2_s, & s\notin U;
\end{cases}
& \mu_{st} = \begin{cases}
\mu^1_{st}, & s\in U,\ t\in U,\\
\mu^2_{st}, & s\notin U,\ t\notin U,\\
\mu^1_{s}(\mu^2_{t})\T, & s\in U,\ t\notin U,\\
\mu^2_{s}(\mu^1_{t})\T, & s\notin U,\ t\in U.
\end{cases}
\end{align}
First, we check that $\mu$ is feasible. We use feasibility of $\mu^1, \mu^2$ and verify that $1\T \mu^1_s (\mu^2_s)\T = (\mu^2)\T$.
\par
Let us now show that $\mu$ is optimal. Let $\varphi$ be relative interior dual solution. By complementarity slackness with $\mu^1$ and $\mu^2$ it must be that $f^{\varphi}_{s}(i) = 0$ whenever $\mu^1_{s}(i)>0$ or $\mu^2_{s}(i)>0$ and the same holds for pairwise terms. We need to care only about the stitching, the pairwise terms in the case $s\in\U$, $t\notin\U$.
Let $O^1_s = \{i \mid \mu^1_{s}(i)>0\}$. Since $|O^1_s| = 1$ and $|O^1_t| = 2$, by feasibility of $\mu$ we have that $\mu_{st}^1(x^1_s,0)>0$ and $\mu_{st}^1(x^1_s,1)>0$. By complementarity, $f^{\varphi}_{st}(x^1_s,0) = f^{\varphi}_{st}(x^1_s,1) = 0$. By construction, $\mu_{st}(1-x^1_s,\cdot) = 0$ and we have that for any $\mu^2_t$ the product $\mu^1_{s}(\mu^2_{t})\T$ satisfies complementarity with $f^{\varphi}$. The remaining case $s\notin\U$, $t\in\U$ is symmetric. Therefore $\mu$ is optimal.
\qed
\end{proof}
It follows that the maximum can be found in polynomial time by trying to fix a variable and check whether there is a feasible solution with such fixation. This is trivial but inefficient. It can be done efficiently by analyzing connected components in the network flow model~\cite[\S 2.3]{Kolmogorov-Rother-07-QBPO-pami}.
\paragraph{Case $K\geq 3$}
To show that for $K\geq 3$ problem {\sc max-wi} is NP-hard 
we notice that
~\eqref{LP} is tight iff there exists $y\in\L$ such that mapping $p\colon \L \mapsto y$ is relaxed-improving. Clearly, this mapping is a (non-unique) solution to {\sc max-wi}. Verifying tightness of~\eqref{LP} is a pairwise constraint satisfaction problem which is NP-hard for $K\geq 3$.
\par
%
%
\subsection{General Properties}\label{sec:properties}
We will now derive some properties of {\sc max-wi/si} problem that will enable our main result -- reduction to a single linear program for subset-to-one maps.
The problem will be gradually reformulated in terms of linear extension $P = [p]$ only.
The constraint $P \in \WI$ is complicating because set $\WI$ is defined with quantifier $(\forall x\in \Lambda)$, see~\eqref{L-improving-conj-a}.
%
However, since $\Lambda$ is polyhedral, this set can be reformulated as a projection of a higher-dimensional polytope: 
\begin{statement}[Dual $\mathbb{W}$]\label{S:WI-dual} Set $\WI$ can be expressed as
\begin{equation}\label{WI-dual}
\{P\colon\Real^\I\to\Real^\I \mid (\exists \varphi\in \Real^m)\ \ f^{\varphi} - P\T f \geq 0 \}.
\end{equation}
\end{statement}
\begin{proof}
Denote $g = (I-P\T)f$. Condition~\eqref{L-LP}, equivalent to~\eqref{L-improving-conj-a}, can be stated for the conic hull of $\Lambda$: 
\begin{equation}\label{L-improving-coni}
\inf_{\mu\in\coni(\Lambda)} \<g, \mu \> \geq 0.
\end{equation}
This is because for any $\mu \in \Lambda$ and any $\alpha \geq 0$ vector $\alpha \mu$ will satisfy RHS of~\eqref{L-improving-conj-a} as well.
Observe that $\coni (\Lambda) = \{\mu \mid A\mu = 0,\ \mu\geq 0\}$ (in the specific representation of the polytope we used we just have to drop the constraint $\mu_0=1$). 
We can write minimization problem in~\eqref{L-improving-coni} and its dual as\\[-5pt]
\begin{equation}\label{LP'}
\begin{array}{rclr}
& \inf \<g,\mu\> &\ \ 
& \max 0\,.\\
&
\setlength{\arraycolsep}{0.2em}
\begin{array}{rl}
A \mu & = 0 \\
\mu & \geq 0
\end{array}
& &
\setlength{\arraycolsep}{0.2em}
\begin{array}{rl}
\varphi & \in \Real^m\\
g - A\T \varphi & \geq 0
\end{array}
\end{array}
\end{equation}
Inequality~\eqref{L-improving-coni} holds iff the primal problem is bounded, and it is bounded iff the dual is feasible, which is the case iff $(\exists\varphi\in\Real^m)$ $(f-A\T\varphi)-P\T f \geq 0$.
\qed
\end{proof}
With this reformulation we can write {\sc max-wi} as
\begin{equation}
\max_{p,\varphi} \sum_{s,i} \leftbb p_s(i){\neq i}\rightbb \tab \mbox{s.t.: } (I-[p]\T)f -A\T \varphi \geq 0.
\end{equation}
Notice, quantifier $(\exists \varphi)$ turned into an extra minimization variable.
To handle the strict case, we would need a similar dual reformulation for the set $\SI$. This set has a more complicated quantifier $(\forall \mu \in \Lambda, P \mu \neq \mu)$. Fortunately, the following reformulation holds for pixel-wise maps:
\begin{statement}[Dual $\mathbb{S}$]\label{S:SI-dual} Let $p\colon \L\to \L$ be pixel-wise. Then
 $[p] \in \SI$ iff $(\exists \epsilon>0)$ $(\exists \varphi\in \Real^m)$
\begin{subequations}\label{SI-dual}
\begin{align}
\label{SI-dual-a}
(\forall s,\,\forall i)\tab& f^{\varphi}_s(i) - f_s(p_s(i)) \geq \epsilon\leftbb p_s(i){\neq}i\rightbb,\\
\label{SI-dual-b}
(\forall st,\,\forall ij)\tab & f^{\varphi}_{st}(i,j) - f_{st}(p_s(i),p_t(j)) \geq 0.
\end{align}
\end{subequations}
\end{statement}
\begin{proof}
Let $h\in\Real^\I$ with components $h_{s}(i) = \leftbb p_s(i){\neq}i\rightbb$, $h_{st}(i,j)=0$. For $\mu\in\Lambda$ there holds $\<h,\mu\> = 0$ iff $[p]\mu=\mu$. Conditions~\eqref{L-improving-conj-b} are equivalent to 
\begin{equation}\label{S-L-h}
(\forall \mu\in\Lambda) \tab \<(I-[p]\T)f,\mu\> \geq \epsilon \<h,\mu\>
\end{equation}
for some $\epsilon>0$.
%
We apply now the same inference as in \Statement{S:WI-dual} for vector $g = f-P\T f - \epsilon h$. It follows that~\eqref{S-L-h} is equivalent to $(\exists \varphi\in\Real^m)$ $(f-A\T\varphi)-P\T f - \epsilon h \geq 0$. 
\qed
\end{proof}
Additionally, the following lemma provides necessary conditions for sets $\WI$, $\SI$. It will help to narrow down the set of maps over which the optimization is carried out. 
\begin{lemma}[Necessary Conditions]\label{necessary-LI} Let $P\colon \Real^\I \to \Real^\I$, $P(\Lambda)\subset\Lambda$ and
$\O = \argmin_{\mu\in\Lambda}\<f,\mu\>$. Then
\begin{enumerate}
\item[(a)]$P\in\WI$ $\Rightarrow$ $P(\O) \subset \O$.
\item[(b)]$P\in\SI$ $\Rightarrow$ $(\forall \mu\in\O)$ $P(\mu) = \mu$.
\end{enumerate}
\end{lemma}
\begin{proof}{\bf (a)\ }Assume $(\exists\mu\in\O)$ $P\mu \in\Lambda\backslash\O$. Then $\<f,P\mu\> > \<f,\mu\>$, therefore $P\notin\WI$.
{\bf (b)\ }Assume $(\exists\mu\in\O)$ $P\mu \neq \mu$. Then $\<f,P\mu\> \geq \<f,\mu\>$ and therefore $P\notin\SI$. 
\qed
\end{proof}

%% file: tex/po_expansion_LP.tex
\par
%
\par
%
%
%
\subsection{Maximum Persistency by LP}
\par
Let us consider the class of maps $\P^{2,y}$, in which mapping $p_{\xi}$ is defined by the indicator variable $\xi\in \{0,1\}^V$.
We will first consider problem \maxwi. 
The constraint $[p_\xi]\in\WI$ in the dual form is still complicated by that $[p_\xi]$ defined by~\eqref{pixel-ext} involves products $\xi_{si}\xi_{tj}$. We are going to linearize these terms by introducing additional variables $\xi_{stij}$. 
Let $\Sigma$ be set the of vectors $\xi$ with components $\xi_{si}$, $\xi_{stij}$ such that\\[-5pt]
\begin{equation}
\begin{aligned}
0 \leq &\xi_{si} \leq 1,\\
\max(0,\xi_{si}+\xi_{tj}-1) \leq & \xi_{stij} \leq \min(\xi_{si},\xi_{tj}).
\end{aligned}
\tag{$\Sigma$}
\end{equation}
If $\xi\in\Sigma$ and all $\xi_{si}$ are integral, there holds $\xi_{stij} = \xi_{si}\xi_{tj}$. 
Set $\Sigma$ is convex, polyhedral.
%
\def\bigeqa{%
(P_\xi \mu)_s(i) = & \sum_{i'} P_{s,ii'} \mu_s(i'), \\
(P_\xi \mu)_{st}(i,j) = &\sum_{i',j'} P_{st,ii',jj'} \mu_{st}(i',j'),
}%
For $\xi\in\Sigma$ we introduce the following corresponding mapping $P_\xi$ by replacing products $\xi_{si}\xi_{tj}$ with $\xi_{stij}$ in~\eqref{pixel-ext}:\\[-12pt]
\begin{subequations}
\begin{align}
\bigeqa
\end{align}
\end{subequations}\\[-25pt]%
\begin{subequations}
\begin{align}
\label{PL-a}
P_{s,ii'} = & \leftbb p_s(i'){=}i \rightbb  = \\
\notag
&\leftbb y_s{=}i \rightbb\xi_{si'}+ \leftbb i'{=}i \rightbb(1-\xi_{si'}),\ \  \ \ \ \  \ \
\end{align}
\begin{align}
\label{PL-b}
P_{st,ii',jj'}  = & \leftbb y_s{=}i \rightbb\leftbb y_t{=}j \rightbb\xi_{sti'j'}\\
\notag
+ &\leftbb i'{=}i \rightbb\leftbb y_t{=}j \rightbb (\xi_{tj'} - \xi_{sti'j'})\\
\notag
+ &\leftbb y_s{=}i \rightbb\leftbb j'{=}j \rightbb (\xi_{si'} - \xi_{sti'j'})\\
\notag
+ &\leftbb i'{=}i \rightbb\leftbb j'{=}j \rightbb (1-\xi_{si'}-\xi_{tj'}+\xi_{sti'j'}).
\end{align}
\end{subequations}
%
%


Mapping $P_\xi$ is linear in $\xi$ and for integer $\xi$ it coincides with $[p_\xi]$. We can now formulate \maxwi as 
the following mixed integer linear program:
\begin{subequations}\label{best L-improving ILP}
\begin{align}
\label{ILP}
\tag{IL1}
& \max\csub{\xi,\varphi} \sum_{s,i}\xi_{si} \\ 
\notag
& (I-P\T_\xi)f - A\T \varphi \geq 0\\
\notag
&\xi\in\Sigma;\ \xi_{si} \in\{0,1\}; \xi_{s y_s} = 0.
\end{align}
\end{subequations}

By relaxing the integrality constraints 
we obtain the linear program
\begin{subequations}
\begin{align}
\label{L1}
\tag{L1}
& \max\csub{\xi,\varphi} \sum_{s,i}\xi_{si} \\ 
\notag
& (I-P\T_\xi)f - A\T \varphi \geq 0\\
\notag
&\xi\in\Sigma;\ \xi_{si} \in [0,1]; \xi_{s y_s} = 0.
\end{align}
\end{subequations}
We will prove in \Theorem{T LP1 tight} that this relaxation is tight and then the program will be simplified by expanding the constraints and optimizing out variables $\xi_{stij}$.
We first need the following lemma.

\begin{lemma}\label{L is closed}
Polytope $\Lambda$ is closed under mapping $P_\xi$, $\xi\in\Sigma$. 
\end{lemma}
\begin{proof}
We verify that $(\forall \mu\in\Lambda)$ $P_\xi \mu \in \Lambda$. Denote $\mu' = P_\xi \mu$.
By constraints of $\Sigma$, all numbers~\eqref{PL-a},~\eqref{PL-b} are non-negative, therefore $\mu'\geq 0$. Constraints $1\T \mu'_s = 1$ hold due to $1\T P_s = 1$. Constraints $1\T \mu'_{st} = (\mu'_{t})\T$ hold due to $\sum_{ii'}P_{st,ii',jj'} = P_{t,jj'}$. \qed
\end{proof}
\begin{theorem}\label{T LP1 tight}
In a solution $(\xi,\varphi)$ to~\eqref{L1} vector $\xi$ is integer.
\end{theorem}
\begin{proof}
We will show that rounding $\xi$ up results in a feasible solution with equal or better objective. Because $\xi$ is feasible to~\eqref{L1}, the mapping $P_\xi$ is $\Lambda$-improving for $f$. Note, at this point, unless $\xi$ is integer it is not guaranteed that $P_\xi(\M) \subset \M$ and we cannot draw any partial optimalities from it, neither $P_\xi$ is guaranteed to be idempotent. 
By \Lemma{L is closed}, $P_\xi (\Lambda) \subset \Lambda$. Therefore
\begin{equation}
(\forall \mu\in\Lambda)\tab \<f, P_\xi P_\xi \mu \> \leq \<f, P_\xi \mu \> \leq \<f, \mu \>.
\end{equation}
It follows that $P_\xi^2 =  P_\xi P_\xi$ is $\Lambda$-improving. Since $P_\xi (\Lambda) \subset \Lambda$, it is also $P_\xi^2 (\Lambda) \subset P_\xi (\Lambda) \subset \Lambda$. Moreover, $P_\xi^2 = P_{\xi'}$ with the following coefficients $\xi'$:
\begin{equation}
\begin{aligned}
\xi'_{si} &  = 1-(1-\xi_{si})^2,\\
\xi'_{stij} & = (1-\xi_{si}-\xi_{tj}+\xi_{stij})^2 -1+\xi'_{si} +\xi'_{tj}.
\end{aligned}
\end{equation}
It can be verified that $\xi'\in\Sigma$.
Let $P_{\xi^*} = \lim_{n\to \infty} (P_\xi)^{2^n}$. Then 
\begin{equation}\label{xi star}
\xi^*_{si}  = \lim_{n\to\infty} 1-(1-\xi_{si})^{2^n} = \leftbb \xi_{si}{>0}\rightbb.
\end{equation}
Since $P_{\xi^*}$ is $\Lambda$-improving, it is feasible to~\eqref{L1}. 
Assume for contradiction that there exist $(s',i')$ such that $0< \xi_{s'i'} < 1$. From~\eqref{xi star} we have $\xi^*_{si} \geq \xi_{si}$ for all $si$ and $\xi^*_{s'i'} > \xi_{s'i'}$. It follows that $\xi^*$ achieves a better objective value, which contradicts the optimality of $\xi$. Therefore $\xi$ is integer. 
\qed
\end{proof}
Since the optimal solution to~\eqref{L1} is integer and unique (as seen from the objective), it is the solution to \maxwi. 
\par
Problem \maxsi can be approached similarly, using the dual definition of $\mathbb{S}$.
The inequalities for pairwise terms~\eqref{SI-dual-b} are linearized exactly the same way as for the weak case, we can write them shortly as
\begin{equation}
((I-P\T_\xi)f - A\T \varphi)_{st}(i,j) \geq 0.
\end{equation}
The inequalities for univariate terms~\eqref{SI-dual-a}, by substituting $p_\xi$ 
can be expressed as
\begin{equation}\label{SI-dual-a-xsi}
(f_{s}(i)-f_s(y_s))\xi_{si} - (A\T \varphi)_{s}(i) \geq \epsilon \xi_{si}\leftbb i\neq y_s\rightbb.
\end{equation}
Since we assume $\xi_{s y_s}=0$, expression~\eqref{SI-dual-a-xsi} is equivalent to
\begin{equation}
(f_{s}(i)-f_s(y_s)-\epsilon)\xi_{si} - (A\T \varphi)_{s}(i) \geq 0,
\end{equation}
\ie, we obtained the same form of constraints as for the weak case, but with slightly modified vector $f$. 
Namely, components $f_s(y_s)$ are increased by $\epsilon$ for all $s$. Let us denote the problem~\eqref{L1} with $\epsilon$-modified vector $f$ as~($\epsilon$-L1). Since the solution $\xi$ to ($\epsilon$-L1) is integer it solves {\sc max-si}.
\par
These solutions can be applied for one or more test labelings $y$. A polynomial algorithm, for example, can iterate over labelings $(y^\alpha \mid \forall s\ y_s = \alpha)$ for $\alpha=0,\dots,K-1$. This algorithm subsumes simple Goldstein's DEE~\cite{Goldstein-94-dee} and the series of Kovtun's weak one-against-all subproblems for candidate labelings $y^\alpha$. Most efficient in practice seems to set $y_s$ to one of the immovable labels by the necessary conditions by \Lemma{necessary-LI}.
This approach in fact allows to solve optimally {\sc max-si} problem for the next class of mappings.
\paragraph{Reduced Linear Program}
We now detail the program~\eqref{L1} in components and simplify it for the practical implementation.
We will assume without loss of generality that $0= f_s(y_s) = f_{st}(y_s,y_t) = f_{st}(i,y_t) = f_{st}(y_s,j)$. If the problem $\hat f$ does not satisfy these constraints, we chose the {\em equivalent} problem $f$ by letting
\begin{align}\label{zerotop-1}
\notag
& f_{st}\ind{\ij}  = \hat f_{st}\ind{\ij}-\hat f_{st}(i,y_t)-\hat f_{st}(y_s,j)+\hat f_{st}(y_s,y_t)\,,\\
\notag
& f_{s}(i)  =  \hat f_{s}(i)- \hat f_{s}(y_s)+\sum\limits_{t\in\N(s)}[ \hat f_{st}(i,y_t)- \hat f_{st}(y_s,y_t)]\,,\\
& f_{0}  = \hat f_{0}+\sum_{st\in\E} \hat f_{st}(y_s,y_t)+\sum_{s\in\V} \hat f_s(y_s)\,.
\end{align}
It can be verified by substitution that $\f=\energy{\hat f}$ and that $(\forall \mu\in\Lambda)$ $\<\hat f,\mu\> = \<f,\mu\>$.
By construction, the optimal $\xi$ for problem~\eqref{best L-improving} does not depend on this transformation.
We have
\begin{equation}\label{proj simplified}
\begin{aligned}
(f-P\T f)_s(i) &= f_{s}(i) \xi_{si},\\
(f-P\T f)_{st}(i,j) &= f_{st}(i,j)\big(\xi_{si} + \xi_{tj} - \xi_{stij}\big).
\end{aligned}
\end{equation}
With these expansions made, the problem~\eqref{L1} expresses as
\begin{subequations}
\begin{align}
\notag
& \max\csub{\xi,\varphi} \sum_{s,i}\xi_{si} \\ 
& (\forall s,i) \ \ f_{s}(i)\xi_{si} + \sum_{t\in \N(s)} \varphi_{st}(i) -\varphi_s \geq 0, \\
\notag
& (\forall st,\ \forall ij) \\
\label{L1-pair-expand}
& f_{st}(i,j)(\xi_{si}+\xi_{tj}- \xi_{stij}) -\varphi_{st}(i)-\varphi_{ts}(j) \geq 0,\\
& \sum_s \varphi_s \geq 0, \\
&\xi\in\Sigma;\ \xi_{si} \in [0,1]; \xi_{s y_s} = 0.
\end{align}
\end{subequations}
We next optimize out variables $\xi_{stij}$. For each $ij$ variable $\xi_{stij}$ is present only in the constraint ~\eqref{L1-pair-expand} and the constraint of the feasible set $\Sigma$, $\max(0,\xi_{si}+\xi_{tj}-1) \leq \xi_{stij} \leq \min(\xi_{si},\xi_{tj})$. Depending on whether $f_{st}(i,j)$ is positive or negative the optimal value for $\xi_{stij}$, which allows the maximum freedom for~\eqref{L1-pair-expand} is either its lower or upper bound, respectively. 
Let $\L_{st}^{+} = \{ij \mid f_{st}(i,j) > 0 \}$ and $\L_{st}^{-} = \{ij \mid f_{st}(i,j)\leq 0 \}$. Substituting the respective bounds into~\eqref{L1-pair-expand} and using identities
\begin{align}
\notag
&\xi_{si}+\xi_{tj} - \max(0,\xi_{si}+\xi_{tj}-1) = \min(\xi_{si}+\xi_{tj},1),\\
&\xi_{si}+\xi_{tj} - \min(\xi_{si},\xi_{tj}) = \max(\xi_{si},\xi_{tj})
\end{align}
we can rewrite constraints~\eqref{L1-pair-expand} as
\begin{align}
\notag
& (\forall st,\ \forall ij\in\L_{st}^{-})\\
& f_{st}(i,j)\max(\xi_{si},\xi_{tj}) -\varphi_{st}(i)-\varphi_{ts}(j) \geq 0,\\
\notag
& (\forall st,\ \forall ij\in\L_{st}^{+})\\
& f_{st}(i,j)\min(\xi_{si}+\xi_{tj},1) -\varphi_{st}(i)-\varphi_{ts}(j) \geq 0.
\end{align}
Finally, by expressing $\min$ and $\max$ as two linear constraints each, we obtain
%
the following representation of the problem~\eqref{L1}:
\begin{subequations}
\begin{align}
\notag
& \max\csub{\xi,\varphi} \sum_{s,i}\xi_{si} \\ 
& (\forall s,i) \ \ f_{s}(i)\xi_{si} + \sum_{t\in \N(s)} \varphi_{st}(i) -\varphi_s \geq 0, \\
\label{ILP-c-neg}
\notag
& (\forall st,\ \forall ij\in\L_{st}^{-})\\
& f_{st}(i,j)\xi_{si} -\varphi_{st}(i)-\varphi_{ts}(j) \geq 0, \\
& f_{st}(i,j)\xi_{tj} -\varphi_{st}(i)-\varphi_{ts}(j) \geq 0; \\
\label{ILP-c-pos}
\notag
& (\forall st,\ \forall ij\in\L_{st}^{+})\\
& f_{st}(i,j)-\varphi_{st}(i)-\varphi_{ts}(j) \geq 0, \\
& f_{st}(i,j)(\xi_{si}+\xi_{tj}) -\varphi_{st}(i)-\varphi_{ts}(j) \geq 0; \\
\notag
& \sum_s \varphi_s \geq 0, \\
\notag
& \xi_{si} \in [0,1]; \xi_{s y_s} = 0.
\end{align}
\end{subequations}
In this form, only variables $\xi_{si}$ remained. On the other hand, the number of constraints has doubled.
%
\paragraph{All-to-One-Unknown} 
Let us consider the class $\P^{1}$, in which map $p_\xi$ is defined by $\xi\in\{0,1\}^\V$ and labeling $y\in \LL$.
Problem \maxwi is NP-hard by our argument above for $K \geq 3$, valid for this class as well. However, we can solve the {\sc max-si} problem combining necessary conditions by \Lemma{necessary-LI} and ($\varepsilon$-L1) problem 
as proposed in~\Algorithm{Alg:s-all-to-one-unknown}.
\begin{algorithm}
	$\mu\in\argmin_{\mu\in\Lambda} \<f,\mu\>$\tcc*{solve~\eqref{LP}}
	For all $s$ if exists $i\in\L_s$ such that $\mu_{s}(i)=1$ then set $y_s = i$, otherwise set $y_s$ arbitrarily\;
	Solve the problem~($\epsilon$-L1) with $y$\;
	\caption{Max Strong all-to-one-unknown\label{Alg:s-all-to-one-unknown}}
\end{algorithm}
Necessary conditions in this case either provide the unique labeling $y_s$ or prove that $p_s$ must be identity. The optimality of the method follows. This algorithm subsumes strict variant of Kovtun's one-against-all auxiliary problem, under an arbitrary choice of a test labeling $\bar y$ and the iterative pruning method~\cite{Swoboda-13}.
%
%
%
%
%

%% file: tex/windowing.tex
\subsection{Windowing}\label{S:windowing-s}
In this section we would like to address large-scale problems, where solving~\eqref{L1} for the full problem may be numerically intractable. We restrict consideration to a local {\em window} $\W\subset \V$,
 fix $p_s(i)=i$ for all $s\notin \W$ and solve for the part of $p$ inside the window. This leads to a reduced problem~\eqref{L1} with variables $\xi_{si}$ and $\varphi_{st}(i)$ only inside the window. But how do we pick a good labeling $y$ for~\eqref{L1}, without solving the full~\eqref{LP}? We propose the following ``local'' necessary conditions. 
%
%
%
\par
\begin{theorem}\label{T:nested} Let $P,Q \colon \Real^\I \to \Real^\I$, $QP=Q$ and $P(\Lambda)\subset \Lambda$. Let 
\begin{equation}\label{test-problem}
\O = \arg\min_{\mu\in\Lambda} \<f,(I-Q)\mu\>
\end{equation}
Then
\begin{enumerate}
\item[(a)] $P\in\WI$ $\Rightarrow$ $P(\O) \subseteq \O$.
\item[(b)] $P\in\SI$ $\Rightarrow$ $(\forall \mu\in\O) P\mu = \mu$.
\end{enumerate}
\end{theorem}
\begin{proof}
{\bf (a)\ }
Assume $P(\O) \nsubseteq \O$. Then there exists $\mu\in \O$ such that $P\mu\in\Lambda\backslash \O$. Then $\<f,(I-Q) \mu\> < \<f, (I-Q) P \mu\>$. 
Equivalently, $0 > \<f,((I-Q) -(I-Q)P ) \mu\>  = \<f,(I-P) \mu\>$. Hence $P\notin \WI$.
\par
{\bf (b)\ }
Assume for contradiction that $P\mu\neq \mu$. Then $\<f, P\mu \> < \<f, \mu \>$. Equivalently, $\<f,((I-Q) -(I-Q)P ) \mu\> > 0$, which implies $\<f,(I-Q) \mu\> > \<f, (I-Q) P \mu\>$, which contradicts $\mu\in\O$.
\qed
\end{proof}
As a corollary, we have the following result. Let $O_s = \{i\in\L_s \mid (\exists\mu \in \O)\ \mu_s(i)>0 \}$. A pixel-wise map $p\colon \L \to \L$ is weakly $\Lambda$-improving only if $p_s(O_s) = O_s$ and strongly-$\Lambda$-improving only if $p_s(i) = i$ for all $i\in O_s$.
\par 
Instead of solving full~\eqref{LP} we solve test problem~\eqref{test-problem} with $Q=[q]$, $q_s(i) = i$ for $s\notin \W$ and $q_s(i) = 0$ for $s\in \W$. Since for any $p$ in the window there holds $q \circ p = q$, the solution to~\eqref{test-problem} identifies the subset of ``immovable'' labels and makes algorithms developed in the previous section applicable.
%
%
%
%
\par
In order to better understand necessary conditions by \Theorem{T:nested} we give the next additional property. For a projection $P\colon \Real^\I \to \Real^I$, its null space corresponds to the dimensions (variables) that become fixed. The larger the null space the more powerful the projection is, because the optimization domain reduces from $\M$ to $P(\M)$. The next lemma clarifies why property $Q P = Q$ was essential in \Theorem{T:nested}.
\begin{lemma}Let $P$ be idempotent. Then $Q P = Q$ \iff $\Null(P) \subseteq \Null(Q)$.
\end{lemma}
\begin{proof}
Let $QP = Q$. Assume $P x = 0$, then $Qx = QP x = Q 0 = 0$ and therefore $x\in\Null(Q)$. In the other direction, let $\Null(P) \subseteq \Null(Q)$. Let $x$ be arbitrary. Since $P^2=P$, we can represent $x$ with the orthogonal sum $x = x_1 +x_2$ with $x_1\in \Null(P)$ and $x_2\in \Null(I-P)$. We have $P x = P x_1 + P x_2 = P x_2 = x_2$. It follows that $QP x = Q x_2$. Since $\Null(P)\subset \Null(Q)$ we have $Q x_1 = 0$ and therefore $Q x = Q x_1 + Q x_2 = Q x_2$. It follows that $QP x = Q x_2 = Q x$ and therefore $QP = Q$.
\qed
\end{proof}
\begin{figure}[!t]
\setlength{\figwidth}{0.5\linewidth}
\setlength{\tabcolsep}{0pt}
\begin{tabular}{cc}
\begin{tabular}{l}
\includegraphics[width=0.86\figwidth]{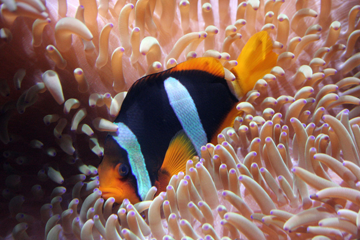}
\end{tabular}&
\begin{tabular}{l}
\includegraphics[width=\figwidth]{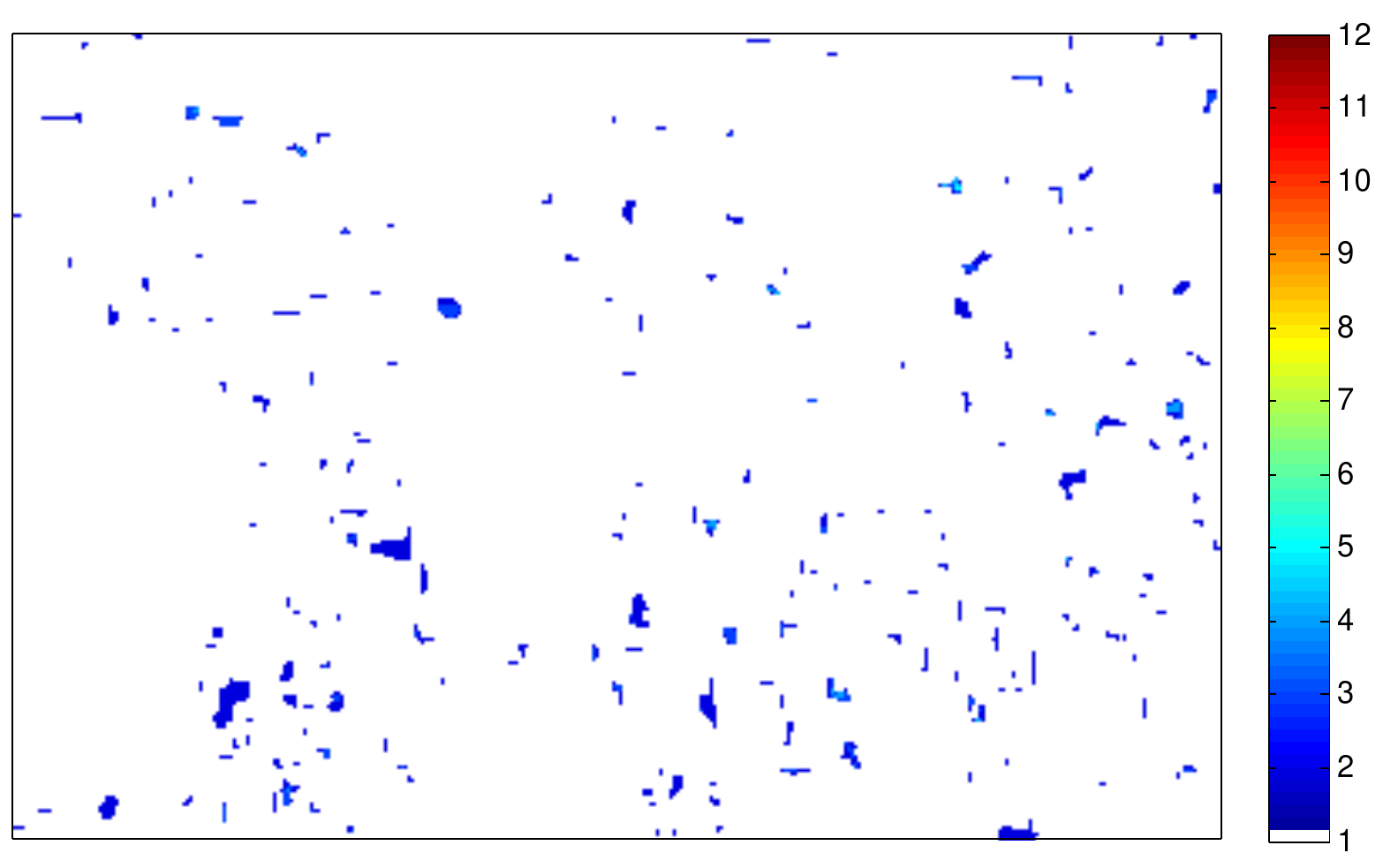}
\end{tabular}\vspace{-1mm}\\
$\wedge$\vspace{-2mm} & \\
\begin{tabular}{l}
\includegraphics[width=\figwidth]{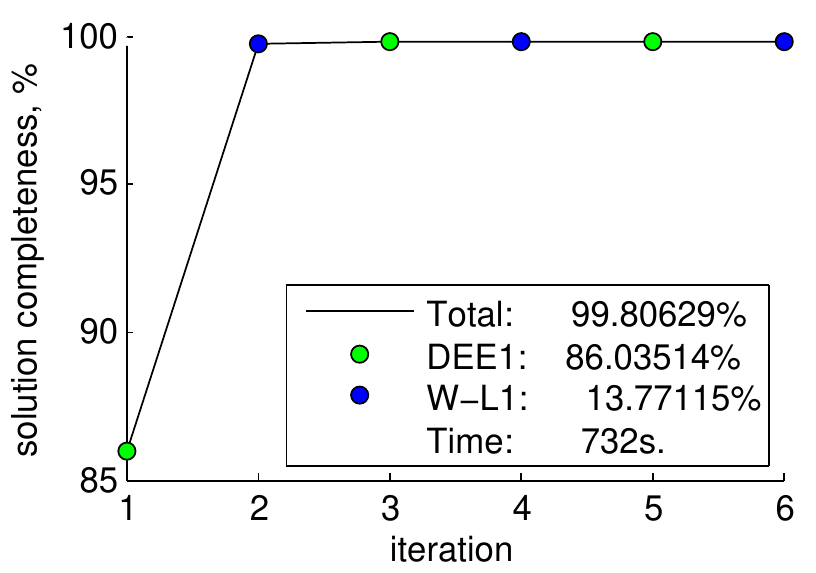}
\end{tabular}&
\begin{tabular}{l}
\includegraphics[width=\figwidth]{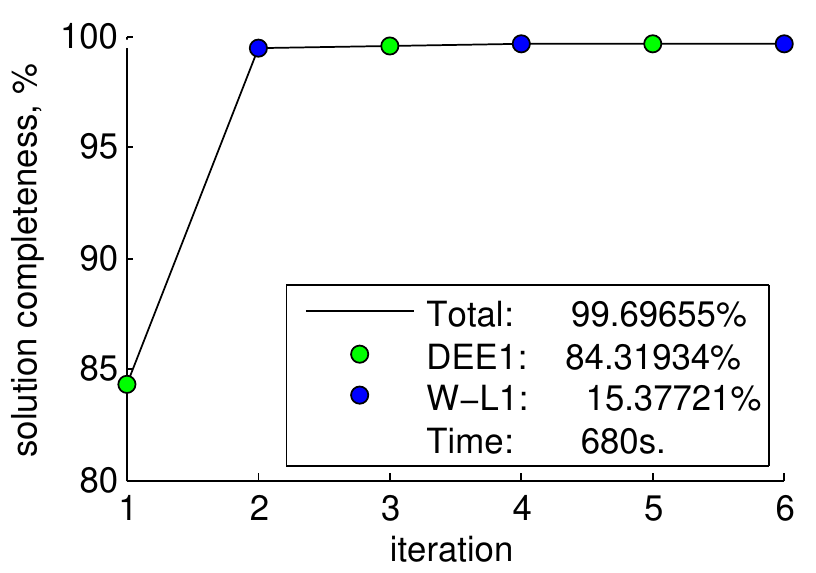}
\end{tabular}\vspace{-2mm}\\ & {$\vee$}\\
\begin{tabular}{l}
\includegraphics[width=0.86\figwidth]{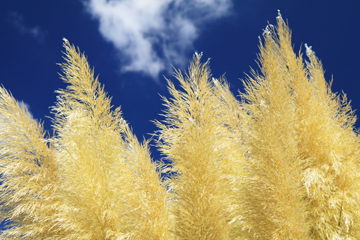}
\end{tabular}&
\begin{tabular}{l}
\includegraphics[width=\figwidth]{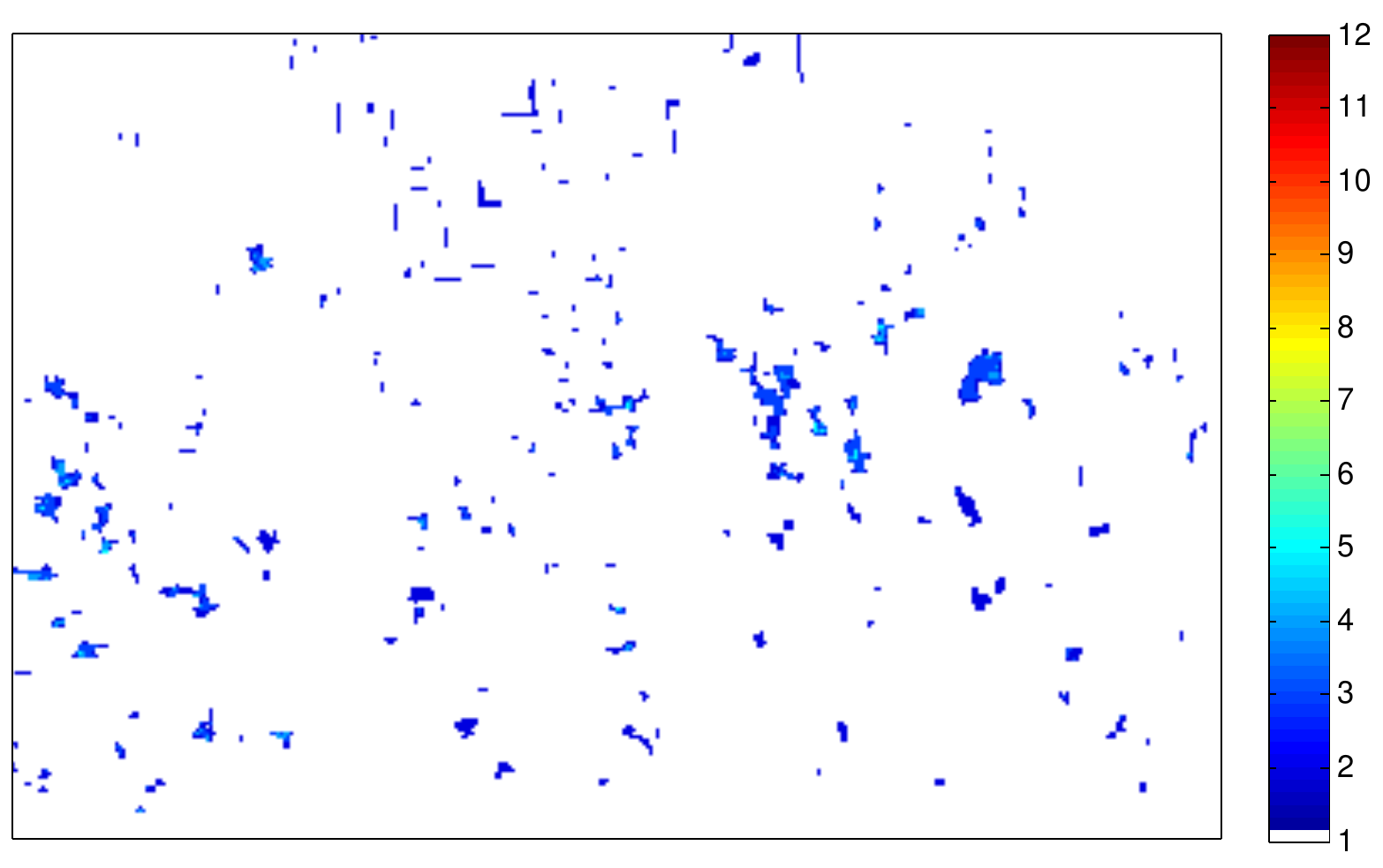}
\end{tabular}\\
\begin{tabular}{l}
\includegraphics[width=0.86\figwidth]{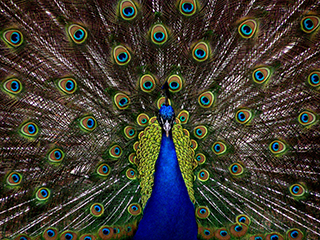}
\end{tabular}&
\begin{tabular}{l}
\includegraphics[width=\figwidth]{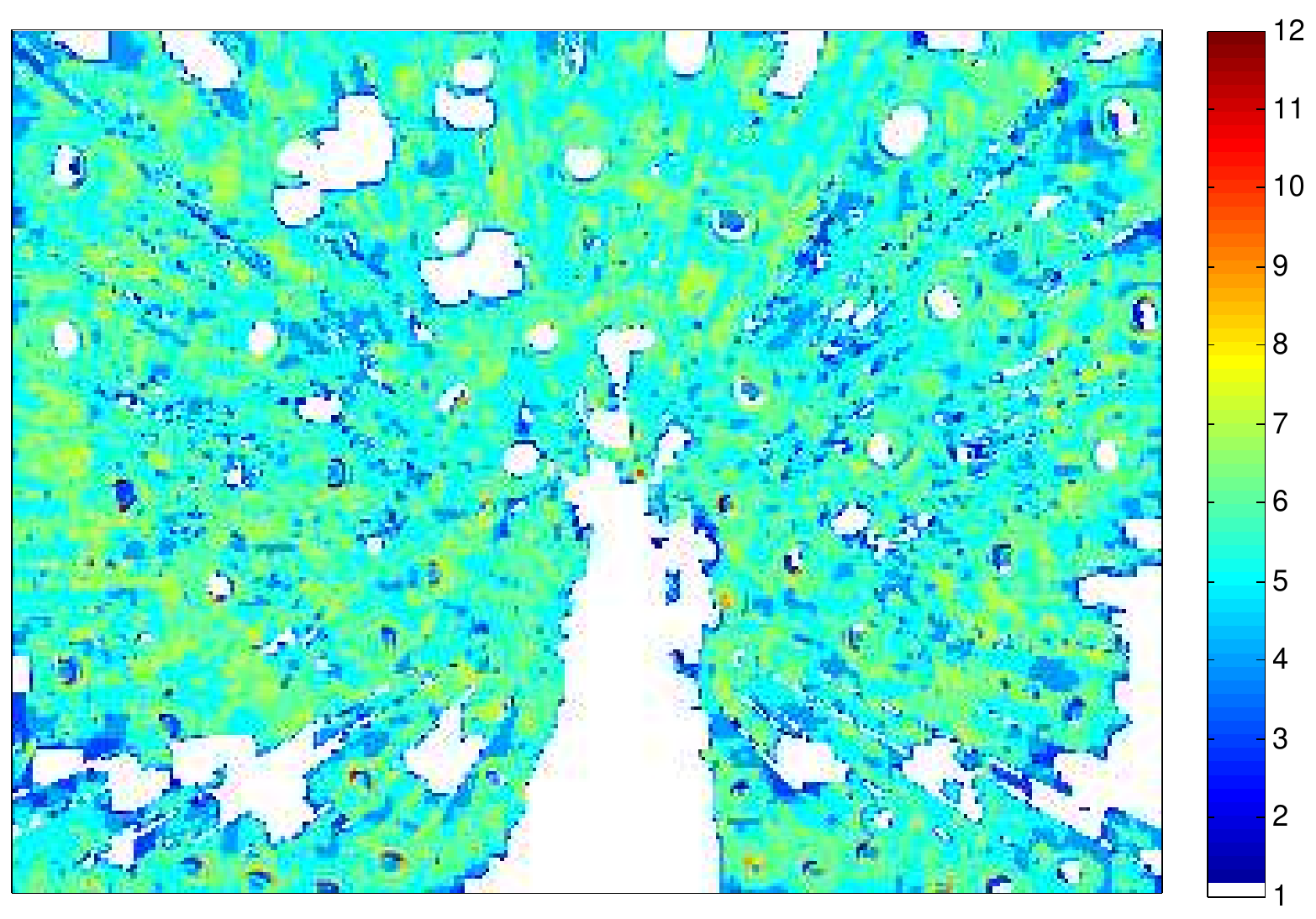}
\end{tabular}\vspace{-1mm}\\
$\wedge$\vspace{-2mm} & \\
\begin{tabular}{l}
\includegraphics[width=\figwidth]{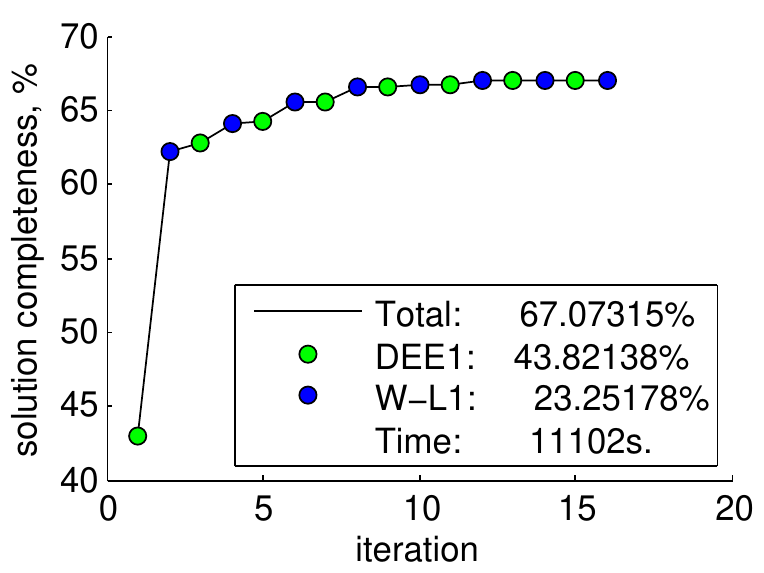}
\end{tabular}
\end{tabular}
\caption{Other instances of {\small \tt color-seg-n4}. For each instance shown: image, reminder of the problem (number of remaining labels), algorithm progress. }\label{f:segm-group2}
\end{figure}

%% file: tex/experiments-rand.tex
\section{Experiments}
\subsection{Random Instances}
We report results on random problems with Potts interactions and full interactions. Both types have unary weights $f_s(i) \sim U[0,\,100]$ (uniformly distributed). Full random energies have pairwise terms $f_{st}(i,j) \sim U[0,\,100]$ and Potts energies have $f_{st}(i,j) = -\gamma_{st}(i)\leftbb i{=}j\rightbb$, where $\gamma_{st}(i)\sim U[0,\, 50]$. All costs are integer to allow for exact verification of correctness. Only instances with non-zero integrality gap \wrt standard LP-relaxation are considered.
For each of the methods in Table~\ref{table1}, we measure {\em solution completeness} as
$\frac{n_{\rm elim}}{|\V|(K-1)}100\%$,
where $n_{\rm elim}$ is the total number of pairs $(s\in\V,i\in\L_s)$ eliminated by the method as non-optimal. The results are shown in Figure~\ref{f:exp-rand}.
%
%
\begin{figure*}[!t]
\includegraphics[width=\linewidth]{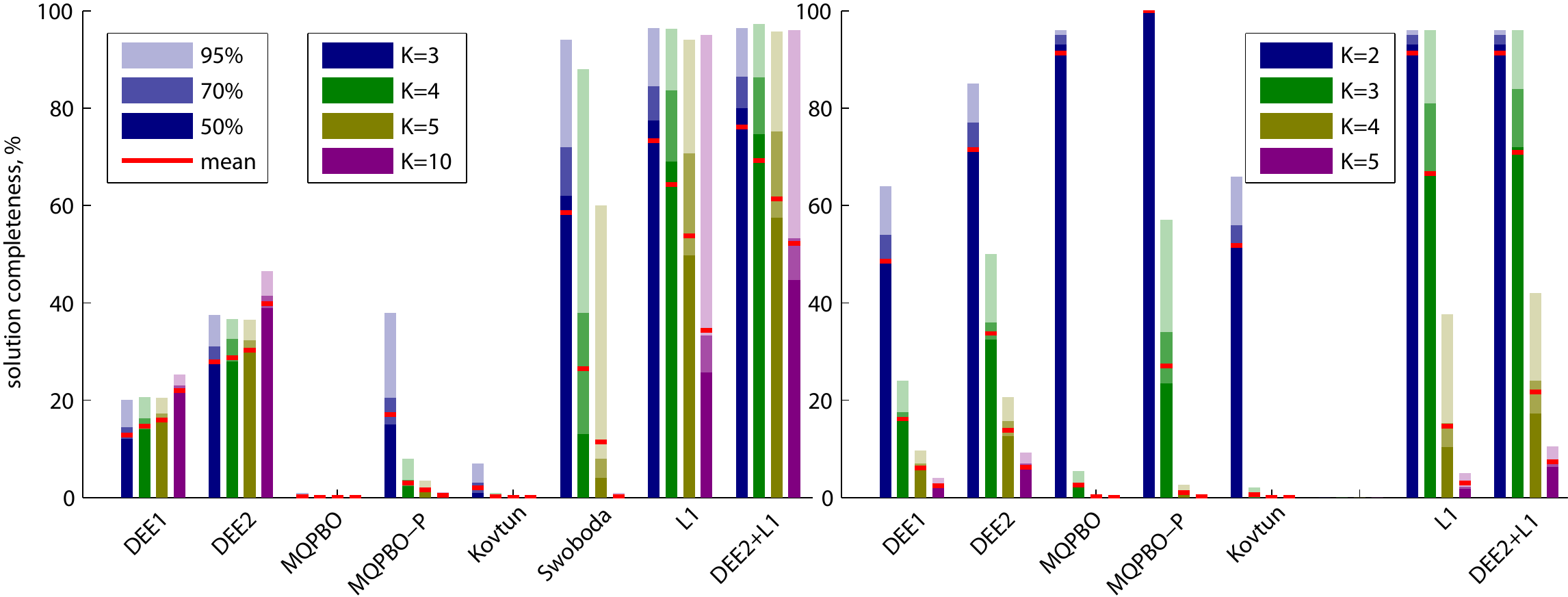}
\caption{Solution completeness by different methods on random instances of size 10x10 pixels, 4-connected. Bars of different shades indicate the portion of the sample under the given solution completeness value (statistics over 100 instances). Left: Potts model, right: full model.}
\label{f:exp-rand}
\end{figure*}
\begin{table*}
{\small
\setlength{\tabcolsep}{3pt}
\renewcommand{\arraystretch}{1.2}
\begin{tabular}{|p{0.072\linewidth}|p{0.905\linewidth}|}
\hline
\scriptsize DEE1 & Goldstein's Simple DEE~\cite{Goldstein-94-dee}: If $f_s(\alpha)-f_s(\beta) +\sum_{t\in\N(s)} \min_{x_t}[f_{st}(\alpha,x_t)-f_{st}(\beta,x_t)] \geq 0$ eliminate $\alpha$. Iterate until no elimination possible.\\
\scriptsize DEE2 & Similar to DEE1, but including also the pairwise condition:
$f_s(\alpha_s)-f_s(\beta_s) + f_t(\alpha_t)-f_t(\beta_t)
+f_{st}(\alpha_{st})-f_{st}(\beta_{st})
+\sum\lsub{t'\in\N(s)\backslash\{t\}} \min_{x_{t'}}[f_{st'}(\alpha_s,x_{t'})-f_{st'}(\beta_s,x_{t'})]
+\sum\lsub{t'\in\N(t)\backslash\{s\}} \min_{x_{t'}}[f_{tt'}(\alpha_t,x_{t'})-f_{tt'}(\beta_t,x_{t'})]
 \geq 0.$\\
\scriptsize MQPBO(-P) & The method of Kohli \etal~\cite{kohli:icml08}. The problem reduced to $\{0,1\}$ variables is solved by QPBO(-P)~\cite{Rother:CVPR07}, where ``-P'' is the variant with probing~\citesuppl{Boros:TR06-probe}. In the options for probing we chose: use weak persistencies, allow all possible directed constraints and dilation=1.\\
\small Kovtun & One-against-all Kovtun's method~\cite{Kovtun-10}. We run a single pass over $\alpha = 1,\dots K$ (test labelings are $(y_s=\alpha\mid s\in\V)$). Labels eliminated in earlier steps are taken correctly into account in the subsequent steps.\\
\small Swoboda & Iterative Pruning method of Swoboda \etal~\cite{Swoboda-13} using CPLEX~\cite{CPLEX} for each iteration. This version is applicable only to Potts model.\\
\scriptsize L1 & The proposed method solving~\eqref{L1} with CPLEX. The test labeling $y$ is selected from the necessary conditions.\\
\scriptsize DEE2+L1 & Sequential application of DEE2 and L1. Note, DEE2 uses condition on pairs which is not covered by the proposed sufficient condition under standard relaxation polytope $\Lambda$. 
\\
\hline
\end{tabular}
}
\caption{List of tested methods.}\label{table1}
\end{table*}

%% file: tex/experiments.tex
\par
Results of all methods that are covered by the proposed sufficient conditions were verified by solving the verification LP~\eqref{L-LP}. For random problems, we also found a global minimum $x^*$ with CPLEX mixed-integer solver (feasible for the size of the problems we used). Methods not verifiable with~\eqref{L-LP} we checked to satisfy $\f(p(x^*)) = \f(x^*)$. 
%
\par
In the case of Potts model, we see that performance of Swoboda \etal~\cite{Swoboda-13} drops quickly with the increase of the number of labels and ours decreases moderately. While the problem difficulty increases, the performance of DEE methods appears to benefit from more labels, which can be explained by the random nature of the problems. Increasing connectivity makes the problem more difficult for all methods, see \Figure{f:potts-3-8}. 
%
%
%
%
%
\begin{figure*}[tbh]
\includegraphics[width=0.5\linewidth]{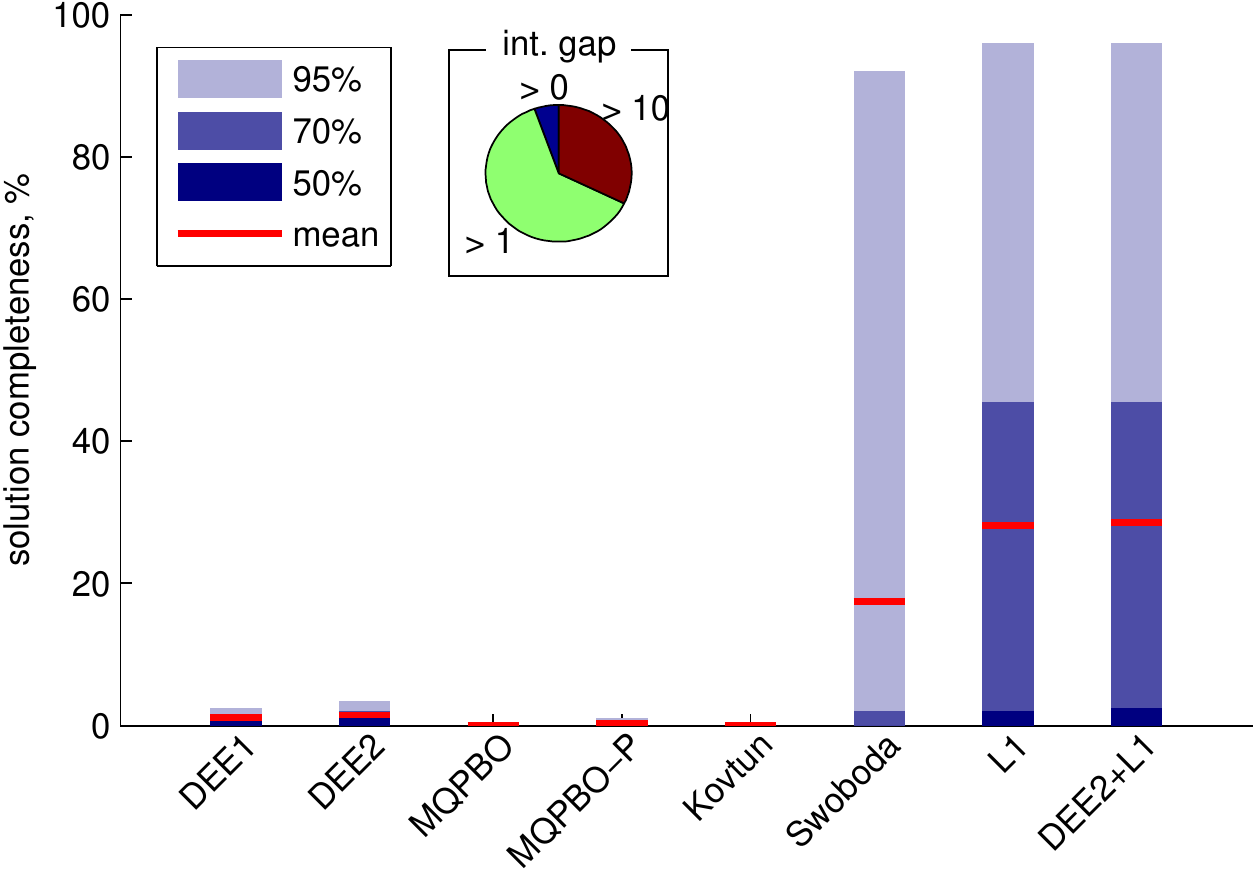}
\includegraphics[width=0.5\linewidth]{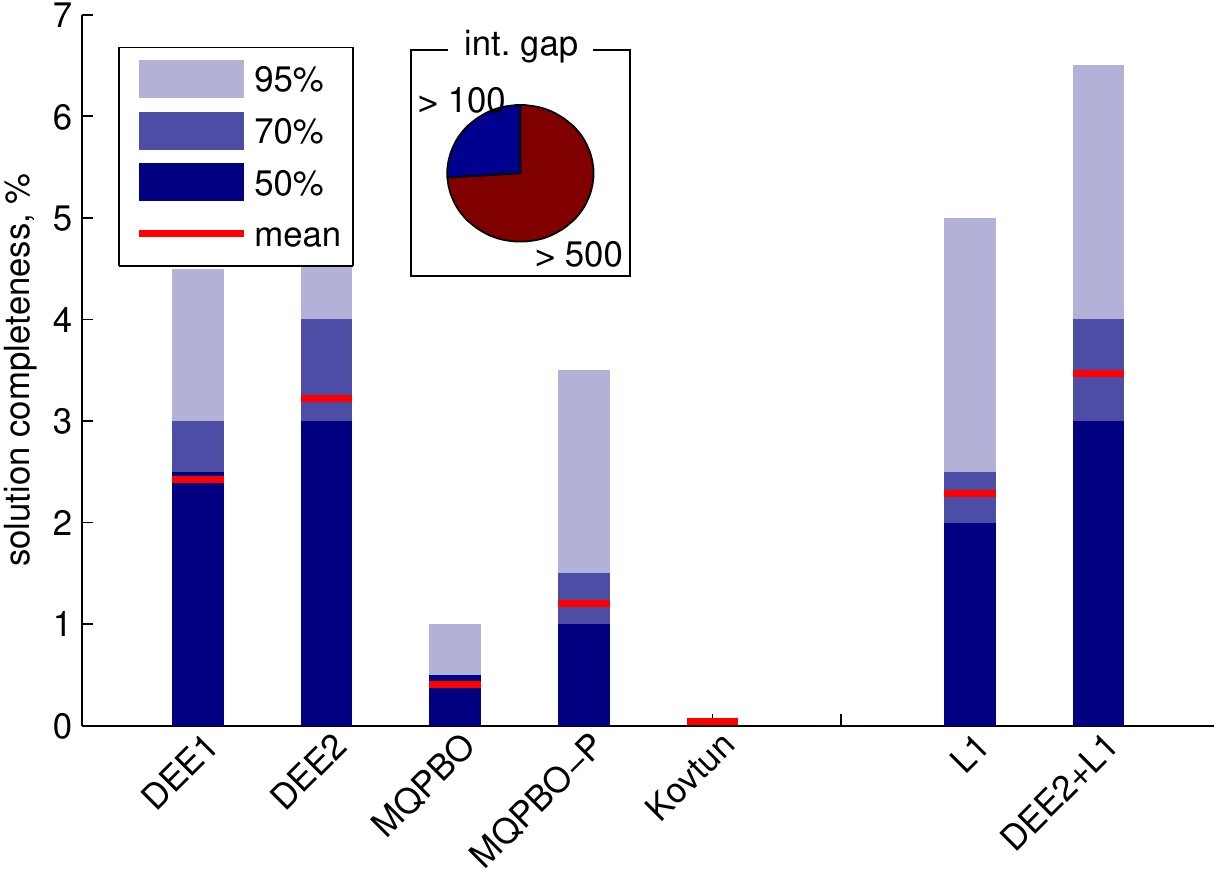}
\caption{Results with 8-connected 10x10 grid, 3 labels. Left: Potts model, right: full model. The pie chart shows distribution of the integrality gap over 100 samples, which indicates problem difficulty. The ``full'' problem is much more difficult (note the different axis scale).}
\label{f:potts-3-8}
\end{figure*}

%
%
%
\begin{figure*}[tbp]
\setlength{\figwidth}{0.23\linewidth}
\begin{tabular}{ccccc}
\begin{tabular}{l}
\includegraphics[height=\figwidth]{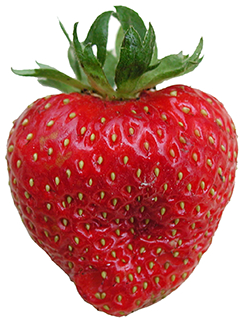}
\end{tabular}&
\begin{tabular}{l}
\includegraphics[height=\figwidth]{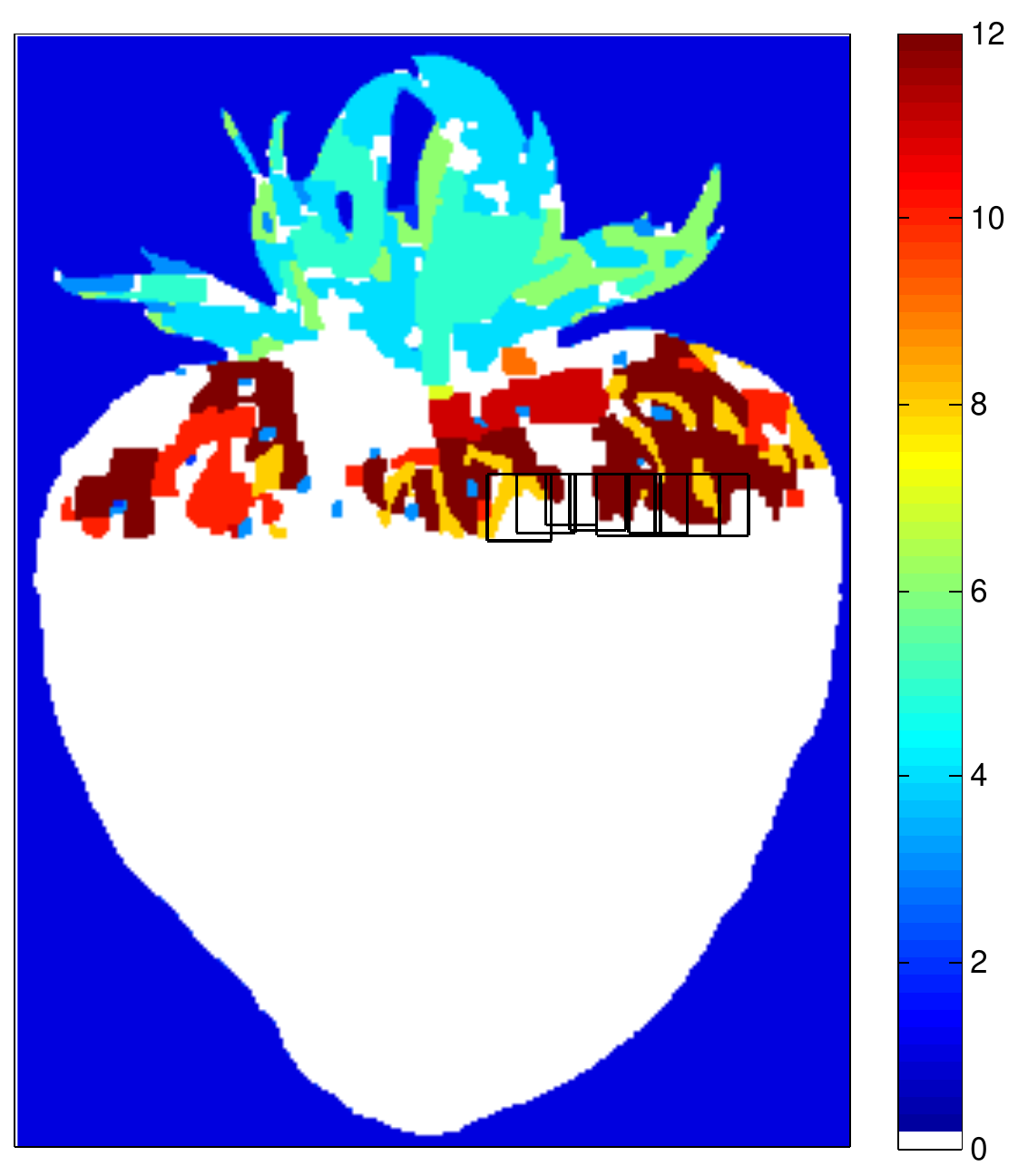}
\end{tabular}&
\begin{tabular}{l}
\includegraphics[height=\figwidth]{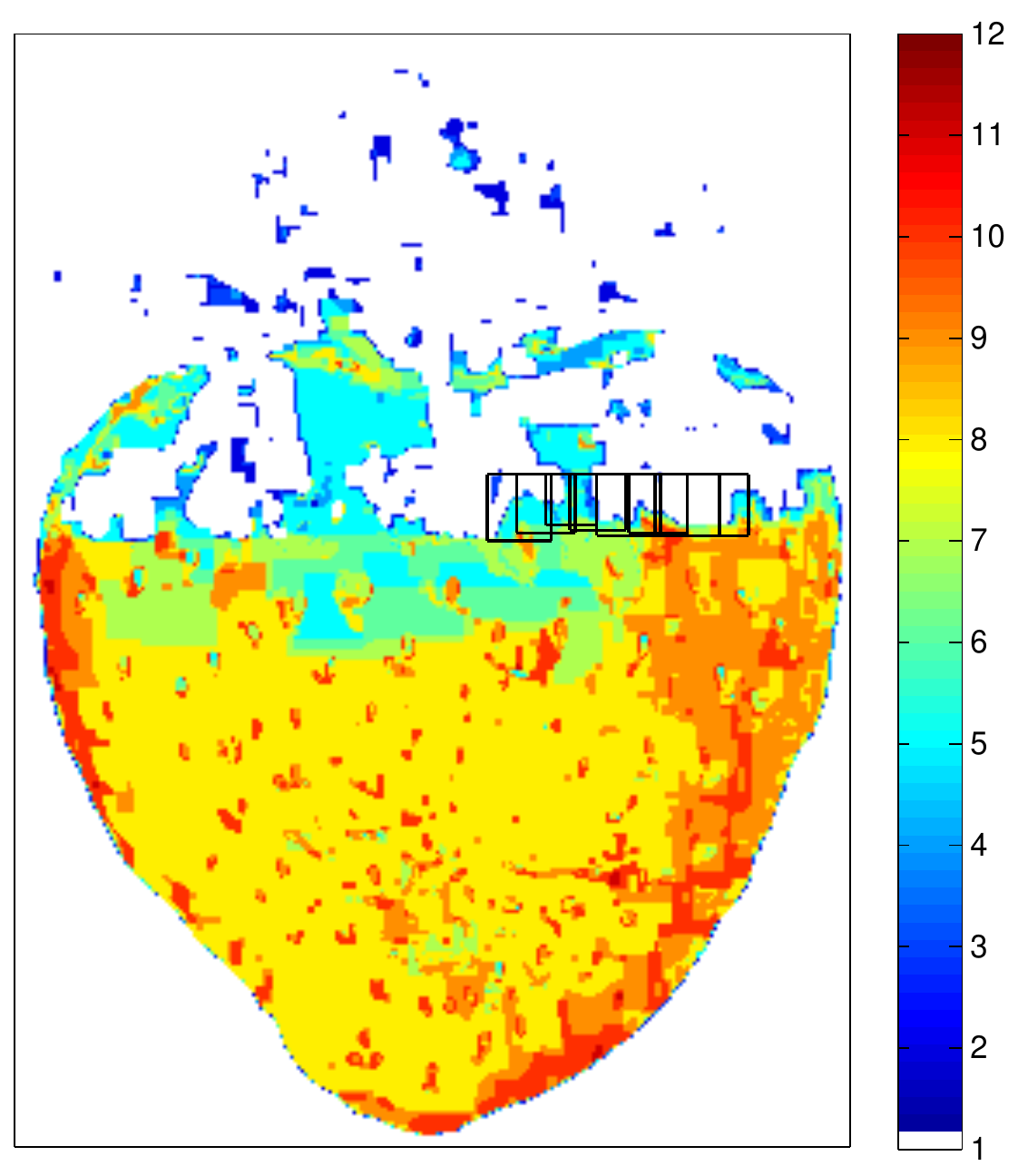}
\end{tabular}&
\begin{tabular}{l}
\includegraphics[height=\figwidth]{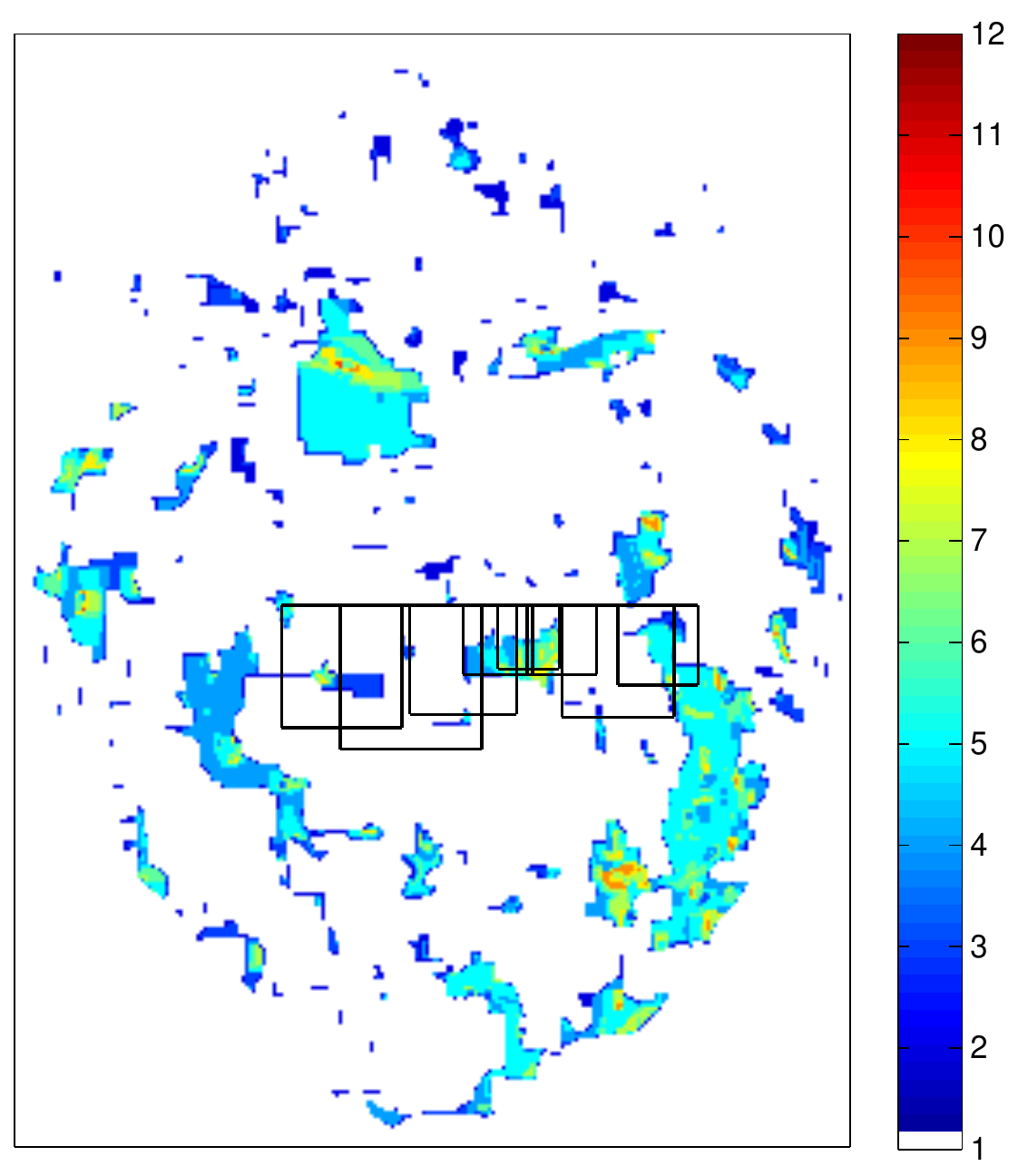}
\end{tabular}&
\begin{tabular}{l}
\includegraphics[height=\figwidth]{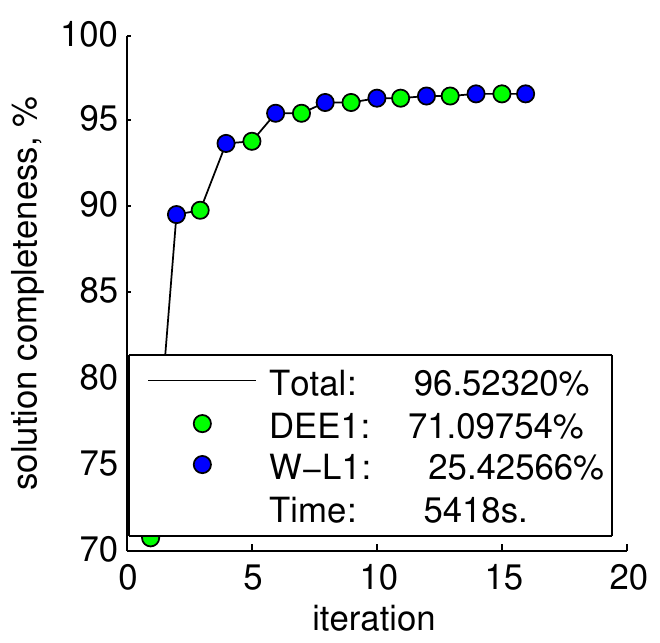}
\end{tabular}\\
(a) & (b) & (c) & (d) & (e)
\end{tabular}
\caption{Windowing method in progress. (a) image corresponding to the instance {\small \tt strawberry-glass-2-small}; (b) Partial labeling for the current reduction: only pixels with a single remaining label are assigned. Black boxes depict the current set of windows selected for application of L1 ( processed in parallel). (c) The number of labels remaining in every pixel on the same iteration as (b). White color indicates fully resolved pixels. (d) The number of remained labels upon termination. Selected windows are larger, but cannot improve the reduction. (e) Algorithm progress during DEE1 iterations (green) and window-L1 iterations (blue).}\label{f:windows}
\end{figure*}
\subsection{Large-Scale Segmentation}
We propose experiments with multiclass image segmentation. We used {\small \verb=color-seg-n4=} instances from~\citesuppl{Kappes-2013-benchmark}, which have 4-12 labels and Potts pairwise interactions. 
Solving LP-relaxation for the whole problem is numerically intractable. We apply the technique described in~\Section{S:windowing-s}. We maintain a global pixel-wise mapping $p\L\to\L$, which defines the current problem reduction. At each step we select a window $\W\subset \V$ such that the problem~\eqref{L1} over the window has no more that $10^4$ variables or constraints (under the current reduction of label sets, $p_s(L_s)$). We find an improving mapping $p'$ from the window subproblem and calculate the composition $p\circ p'$. We can process several overlapping windows in parallel, taking a composition of the mappings in the end. The result might depend on the order of composition, but any order corresponds to a correct weak partial optimality. An example of windows selected for processing instance {\small \verb=crops-small=} are shown in \Figure{f:windows}; Before each scan with local windows we perform simple DEE step, this step makes a big initial reduction for some of these problems, and our method works on the more difficult reminder. On some other problems DEE step is of almost of no help, ({\small \verb=fourcolors=}, fourth in \Figure{f:segm-group1}). By this technique we demonstrate how a nearly complete solution can be found for large instances, by considering always only a part of problem at a time. We see that the {\em reminder of the problem} (the final reduced problem) often decouples in several small independent components, that are feasible to ,\eg, ILP methods. These experiments are a proof-of-concept, we definitely need to develop methods further for a practical implementation. Results are shown in \Figure{f:windows}-\ref{f:segm-group2}. Note, for some of these instances method of Swoboda \etal~\cite{Swoboda-13} identifies a more complete solution, despite we claimed to generalize it. They are using a suboptimal LP solver, but applying it globally to the whole problem. It is likely that our results can be improved by picking the windows more accurately.
\par
One interesting consequence of the windowing method is that it can be applied also with methods of Kovtun~\cite{Kovtun-10}, Swoboda \etal~\cite{Swoboda-13}, MQPBO~\cite{kohli:icml08} and, in fact, any other method that constructs a pixel-wise improving mapping. Because MQPBO constructs a flow network on the graph with $K*|\V|$ nodes and $K^2|\E|$ edges, it was reported as intractable for several vision problems~\citesuppl{Kappes-2013-benchmark}. The proposed windowing technique can remove this limitation.
%
%
%
%
\begin{figure}[!t]
\setlength{\figwidth}{0.4715\linewidth}
\setlength{\tabcolsep}{0pt}
\begin{tabular}{cc}
\begin{tabular}{l}
\includegraphics[width=0.86\figwidth]{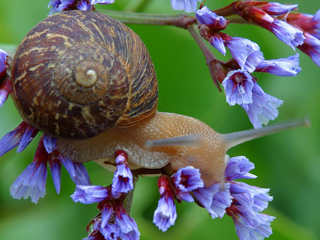}
\end{tabular}&
\begin{tabular}{l}
\includegraphics[width=\figwidth]{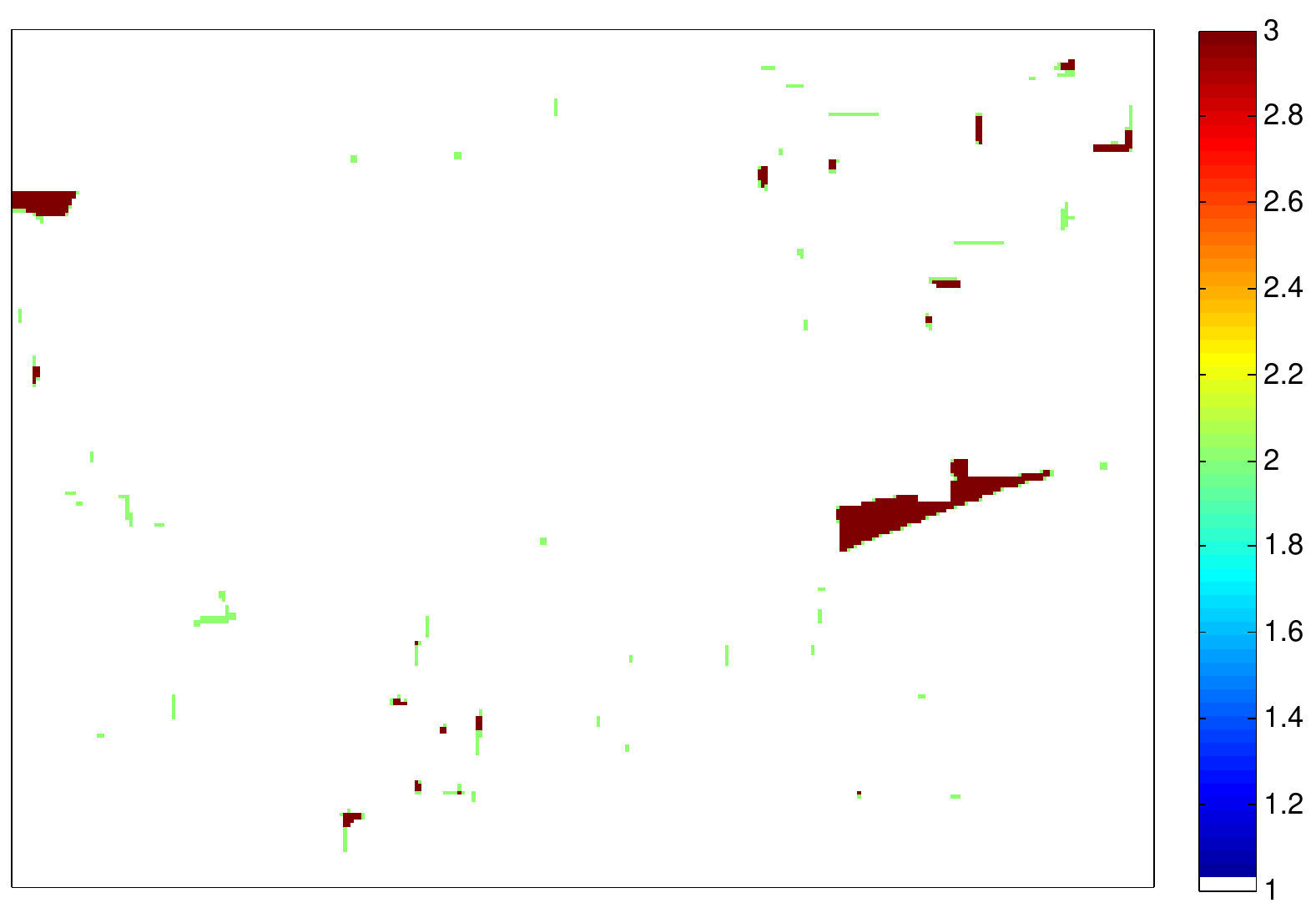}
\end{tabular}\\
$\wedge$\vspace{-3mm}
 & \\
\begin{tabular}{c}
\includegraphics[width=\figwidth]{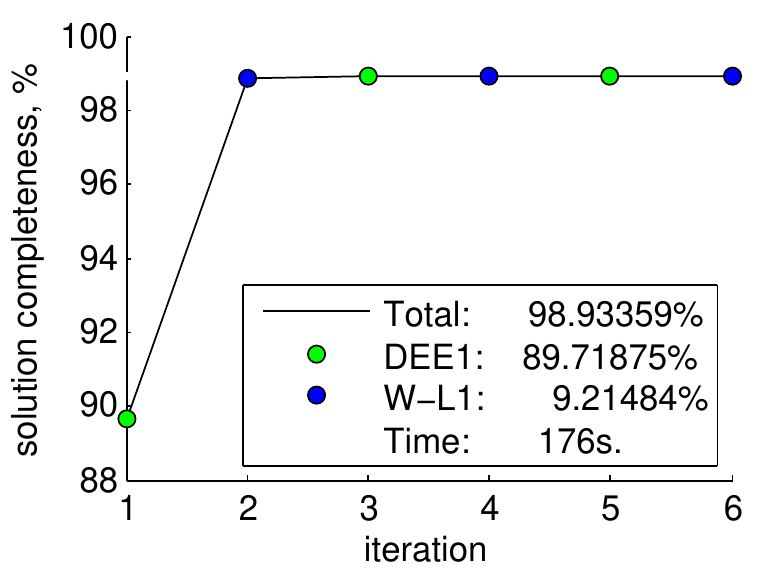}
\end{tabular}&
\begin{tabular}{c}
\includegraphics[width=\figwidth]{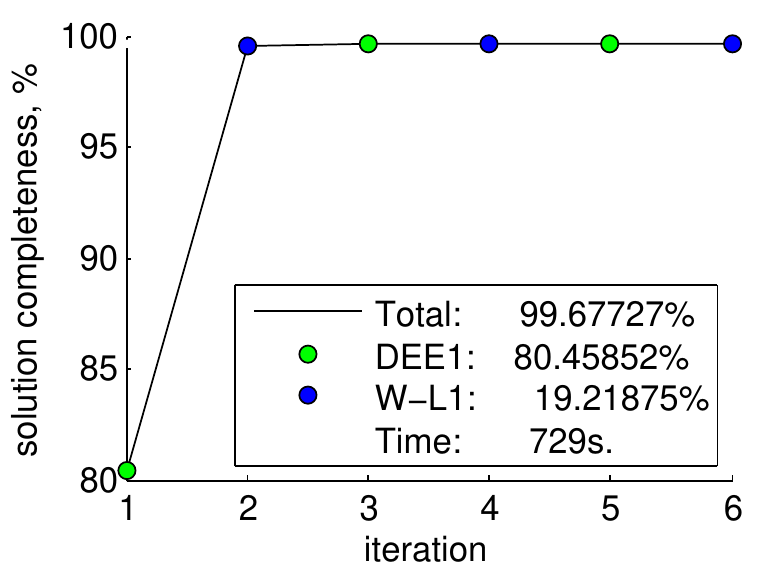}\\
\end{tabular}%
\vspace{-2mm}
\\
& $\vee$
\\
\begin{tabular}{l}
\includegraphics[width=0.86\figwidth]{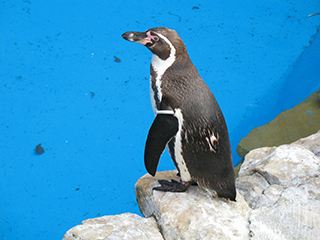}
\end{tabular}&
\begin{tabular}{l}
\includegraphics[width=\figwidth]{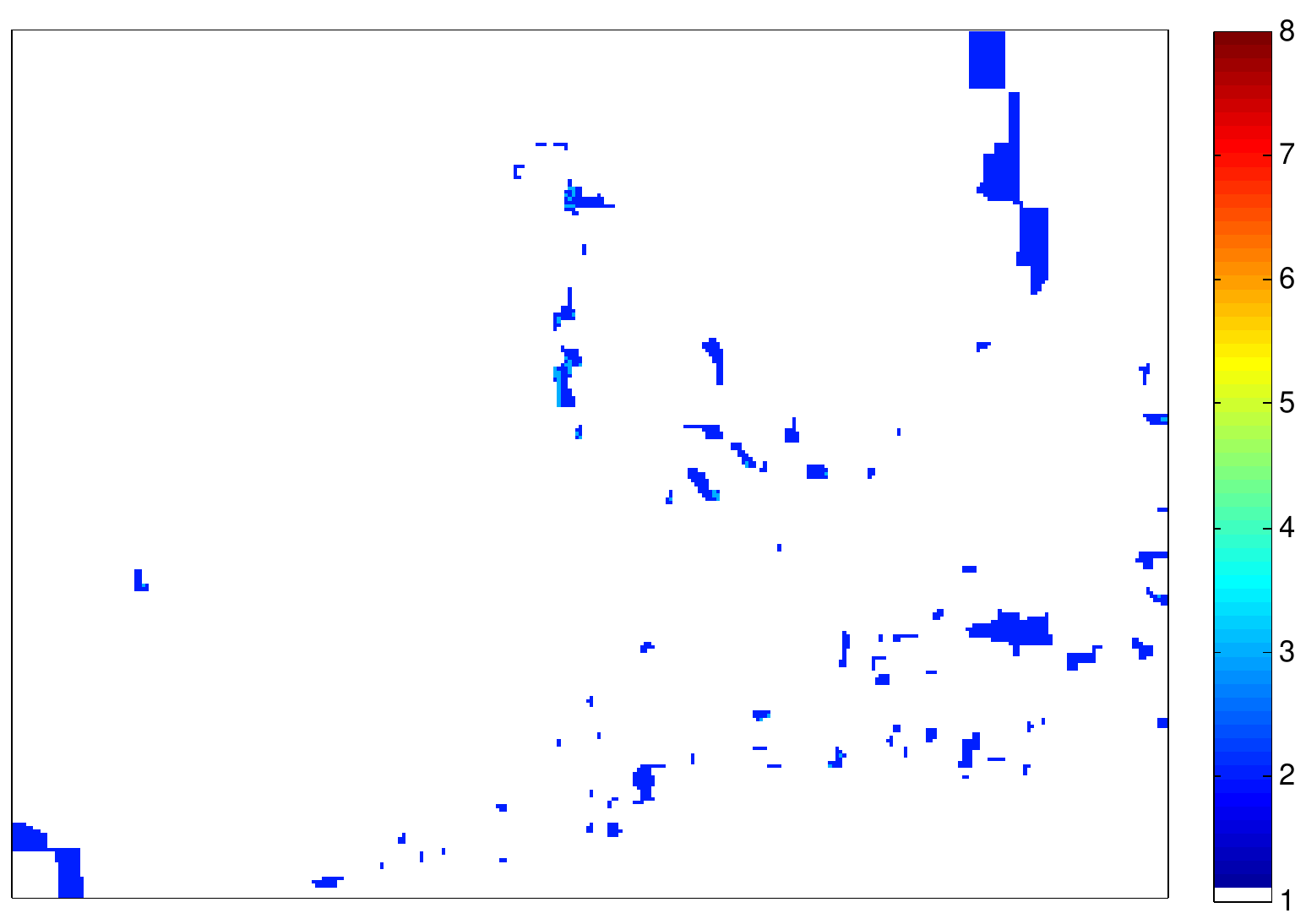}
\end{tabular}\\
\begin{tabular}{l}
\includegraphics[width=0.86\figwidth]{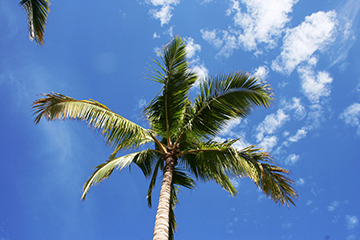}
\end{tabular}&
\begin{tabular}{l}
\includegraphics[width=\figwidth]{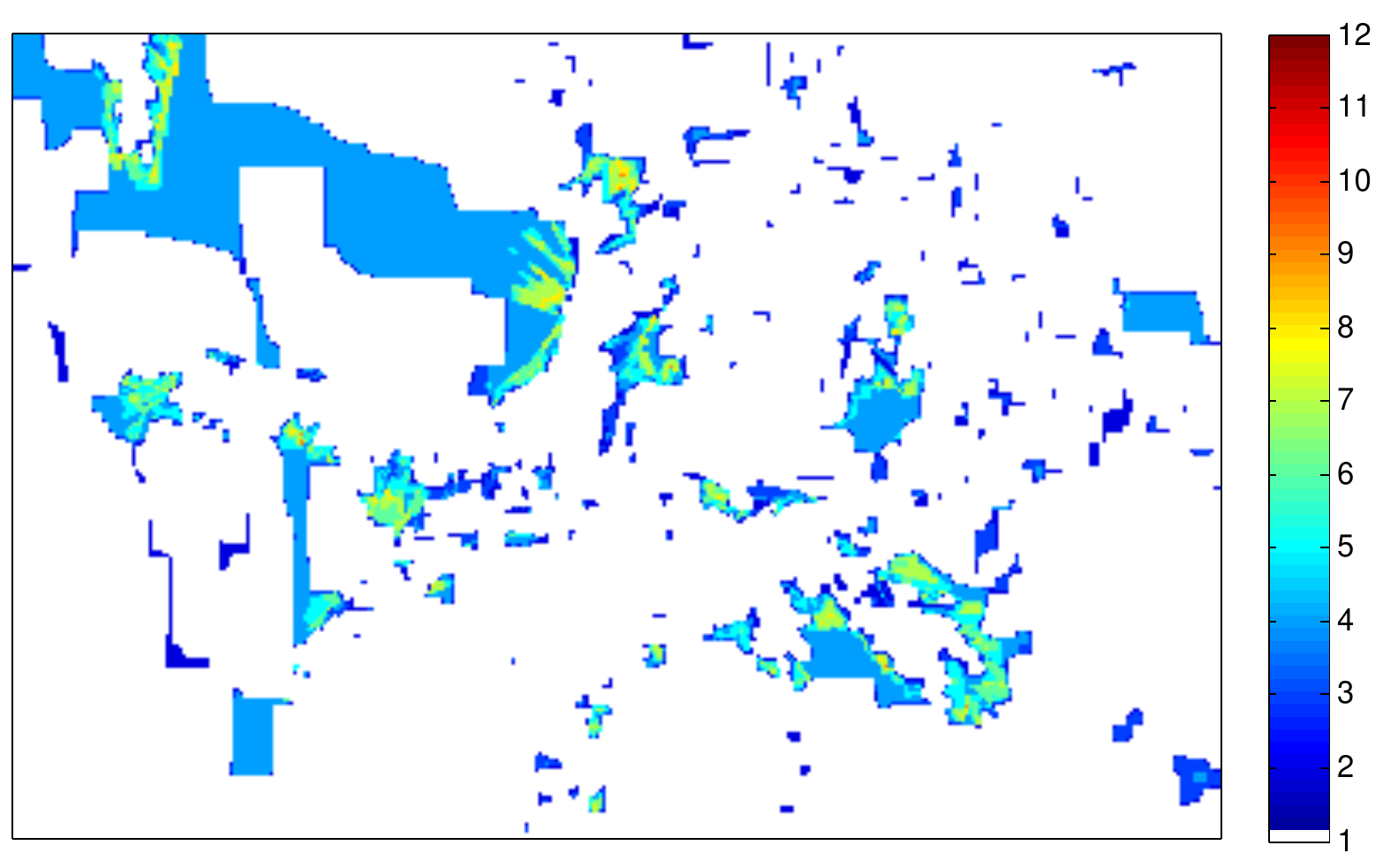}
\end{tabular} \\
{$\wedge$}\vspace{-3mm} & \\
\begin{tabular}{l}
\includegraphics[width=\figwidth]{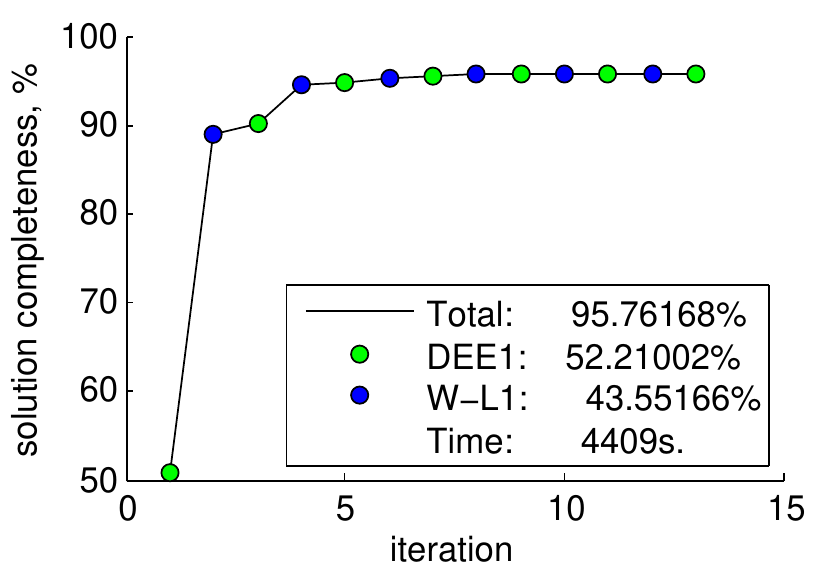}
\end{tabular} &
\begin{tabular}{l}
\includegraphics[width=\figwidth]{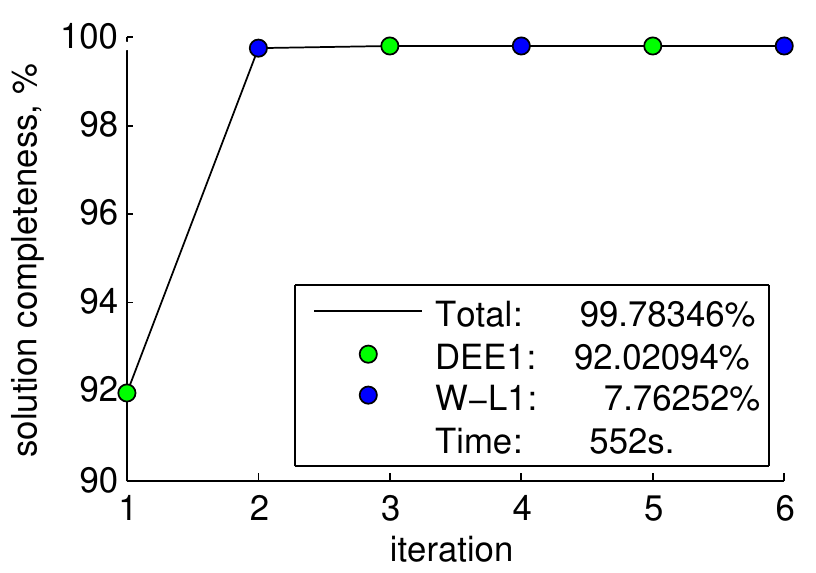}
\end{tabular}\vspace{-2mm}\\
& {$\vee$}\\
\begin{tabular}{l}
\includegraphics[width=0.86\figwidth]{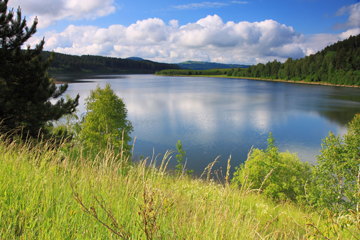}
\end{tabular}&
\begin{tabular}{l}
\includegraphics[width=\figwidth]{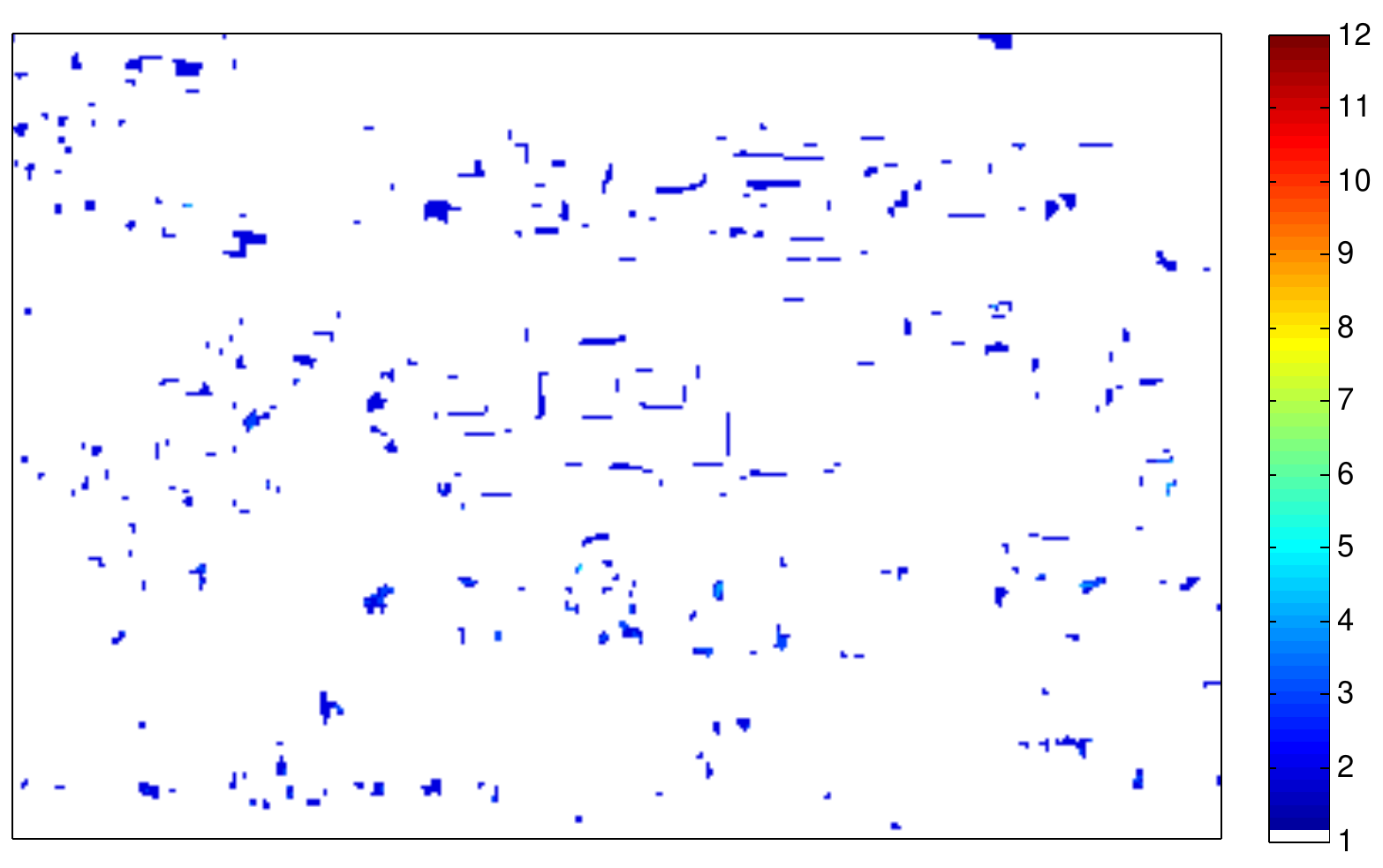}
\end{tabular}\\
\begin{tabular}{cc}
\begin{tabular}{l}
\includegraphics[width=0.78\figwidth]{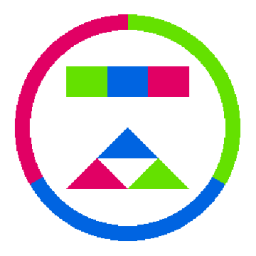}
\end{tabular}%
& $<$\\
\end{tabular}
&
\begin{tabular}{l}
\includegraphics[width=\figwidth]{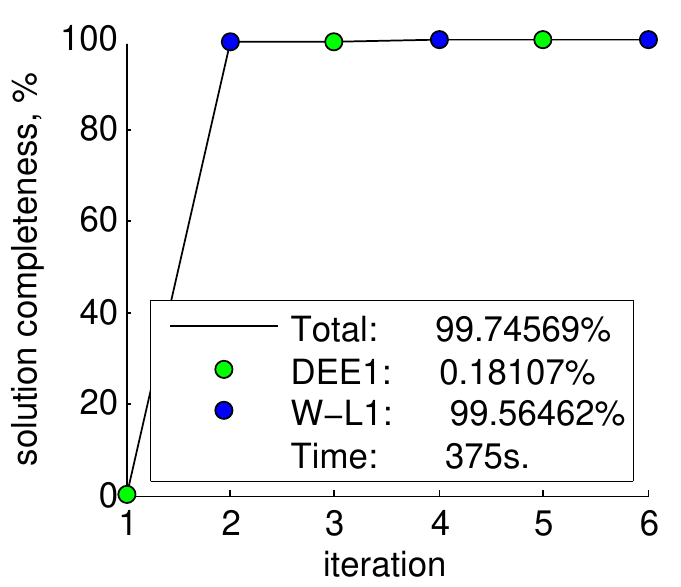}
\end{tabular}%
\end{tabular}
\caption{Other instances of {\small \tt color-seg-n4}.}\label{f:segm-group1}
\end{figure}
%
%
%
%
%
%
%
%
%
%
%

%

%% file: tex/conclusion.tex
%
\section*{Conclusion}\label{S:conclusion}
\addcontentsline{toc}{section}{Conclusion}
We have identified a common mechanism of improving mappings that works in different methods for partial optimality and proposed how to obtain more general optimality guarantees from a given linear relaxation. It leads to a coherent and short description of several methods and analysis of their common properties. From necessary conditions by \Lemma{necessary-LI} it follows that all the methods reviewed in~\Section{S:unification} (as well as the proposed method) cannot be used to tighten the LP-relaxation, they can only simplify it in some cases. While our algorithms work for a restricted class of mappings, many previous methods are based on more narrow classes and use less powerful sufficient conditions. We therefore have a theoretical guarantee to improve over these methods and we have verified on difficult random problems that the improvement is significant. 
\par
The difference between “week” and “strict” conditions may seem unimportant in practice and was often neglected in the previous work. However, the class of mappings for which the maximum persistency problem is polynomially solvable is larger for strict conditions. Therefore, the difference is important for developing algorithms and for the theoretical comparison of different methods. We believe it is also essential for clarity and completeness to keep track of both. Moreover, it may be useful in practice to have a threshold $\varepsilon$, below which (\eg, due to limited numerical or data accuracy) the optimal assignment is not reliable, \cf our strict conditions. 
\par
We also proposed how our method can be applied to large-scale problems on sparse graphs, where solving full-size~\eqref{L1} is numerically intractable. 
We can solve constrained variants of {\sc max-wi}/{\sc max-si}, where the mapping is chosen only inside a window $\W\subset \V$. This leads to linear programs of a smaller size and allows to test the method on vision problems. The windowing technique allows to apply previous methods by parts as well.
%
\par
Our approach is readily generalizable to higher order energies. It would be sufficient to augment the embedding $\delta$ with more components in order to obtain a tighter relaxation and a tighter partial optimality condition (the local polytope $\Lambda$ would be defined as $\aff (\M) \cap \Real^\I_+$).
%
%
%
%
%

%% file: improving_mappings-TR.bbl
\begin{thebibliography}{10}\itemsep=-1pt

\bibitem{CPLEX}
{IBM ILOG CPLEX} {O}ptimization {S}tudio {V}12.5.

\bibitem{Adams:1998}
W.~P. Adams, J.~B. Lassiter, and H.~D. Sherali.
\newblock Persistency in 0-1 polynomial programming.
\newblock {\em Mathematics of Operations Research}, 23(2):359--389, Feb. 1998.

\bibitem{Boros:TR91-maxflow}
E.~Boros, P.~L. Hammer, and X.~Sun.
\newblock Network flows and minimization of quadratic pseudo-{B}oolean
  functions.
\newblock Technical Report {RRR} 17-1991, RUTCOR, May 1991.

\bibitem{Boros:TR06-probe}
E.~Boros, P.~L. Hammer, and G.~Tavares.
\newblock Preprocessing of unconstrained quadratic binary optimization.
\newblock Technical Report {RRR} 10-2006, RUTCOR, Apr. 2006.

\bibitem{Desmet-92-dee}
J.~{Desmet}, M.~D. {Maeyer}, B.~{Hazes}, and I.~{Lasters}.
\newblock {The dead-end elimination theorem and its use in protein side-chain
  positioning}.
\newblock {\em Nature}, 356:539--542, 1992.

\bibitem{Fix-11}
A.~Fix, A.~Gruber, E.~Boros, and R.~Zabih.
\newblock A graph cut algorithm for higher-order {M}arkov random fields.
\newblock In {\em ICCV}, pages 1020 --1027, 2011.

\bibitem{Georgiev-06-dee}
I.~Georgiev, R.~H. Lilien, and B.~R. Donald.
\newblock Improved pruning algorithms and divide-and-conquer strategies for
  dead-end elimination, with application to protein design.
\newblock {\em Bioinformatics}, 22:174--183, July 2006.

\bibitem{Goldstein-94-dee}
R.~F. Goldstein.
\newblock Efficient rotamer elimination applied to protein side-chains and
  related spin glasses.
\newblock {\em Biophysical Journal}, 66(5):1335--1340, May 1994.

\bibitem{Gridchyn-13}
I.~Gridchyn and V.~Kolmogorov.
\newblock Potts model, parametric maxflow and k-submodular functions.
\newblock In {\em ICCV}, 2013.

\bibitem{Hammer-84-roof-duality}
P.~Hammer, P.~Hansen, and B.~Simeone.
\newblock Roof duality, complementation and persistency in quadratic 0-1
  optimization.
\newblock {\em Mathematical Programming}, pages 121--155, 1984.

\bibitem{Ishikawa-11-transform}
H.~Ishikawa.
\newblock Transformation of general binary {MRF} minimization to the
  first-order case.
\newblock {\em PAMI}, 33(6):1234--1249, 2011.

\bibitem{KahlS12}
F.~Kahl and P.~Strandmark.
\newblock Generalized roof duality.
\newblock {\em Discrete Applied Mathematics}, 160(16-17):2419--2434, 2012.

\bibitem{Kappes-2013-benchmark}
J.~H. Kappes, B.~Andres, F.~A. Hamprecht, C.~Schn{\"o}rr, S.~Nowozin, D.~Batra,
  S.~Kim, B.~X. Kausler, J.~Lellmann, N.~Komodakis, and C.~Rother.
\newblock A comparative study of modern inference techniques for discrete
  energy minimization problem.
\newblock In {\em CVPR}, 2013.

\bibitem{kohli:icml08}
P.~Kohli, A.~Shekhovtsov, C.~Rother, V.~Kolmogorov, and P.~Torr.
\newblock On partial optimality in multi-label {MRF}s.
\newblock In {\em ICML}, pages 480--487, 2008.

\bibitem{Kolmogorov10-bisub}
V.~Kolmogorov.
\newblock Generalized roof duality and bisubmodular functions.
\newblock In {\em NIPS}, pages 1144--1152, 2010.

\bibitem{Kolmogorov12-bisub}
V.~Kolmogorov.
\newblock Generalized roof duality and bisubmodular functions.
\newblock {\em Discrete Applied Mathematics}, 160(4-5):416--426, 2012.

\bibitem{Kolmogorov-Rother-07-QBPO-pami}
V.~Kolmogorov and C.~Rother.
\newblock Minimizing nonsubmodular functions with graph cuts -- a review.
\newblock {\em PAMI}, 29(7):1274--1279, 2007.

\bibitem{Kovtun03}
I.~Kovtun.
\newblock {Partial optimal labeling search for a {NP}-hard subclass of (max, +)
  problems}.
\newblock In {\em {DAGM}-Symposium}, pages 402--409, 2003.

\bibitem{Kovtun-10}
I.~Kovtun.
\newblock Sufficient condition for partial optimality for (max, +) labeling
  problems and its usage.
\newblock {\em Control Systems and Computers}, (2), 2011.
\newblock Special issue.

\bibitem{Lasters-95-dee}
I.~Lasters, M.~De~Maeyer, and J.~Desmet.
\newblock Enhanced dead-end elimination in the search for the global minimum
  energy conformation of a collection of protein side chains.
\newblock {\em Protein Engineering}, 8(8):815--22, 1995.

\bibitem{Lu-Williams-1987}
S.~H. Lu and A.~C. Williams.
\newblock Roof duality for polynomial 0-1 optimization.
\newblock {\em Mathematical Programming}, 37(3):357--360, May 1987.

\bibitem{Nemhauser-75}
G.~L. Nemhauser and L.~E. Trotter, Jr.
\newblock Vertex packings: Structural properties and algorithms.
\newblock {\em Mathematical Programming}, 8:232--248, 1975.

\bibitem{Picard-77}
J.-C. Picard and M.~Queyranne.
\newblock On the integer-valued variables in the linear vertex packing problem.
\newblock {\em Mathematical Programming}, 12(1):97--101, 1977.

\bibitem{Pierce-2000-dee}
N.~A. Pierce, J.~A. Spriet, J.~Desmet, and S.~L. Mayo.
\newblock Conformational splitting: A more powerful criterion for dead-end
  elimination.
\newblock {\em Journal of Computational Chemistry}, 21(11):999--1009, 2000.

\bibitem{Rother:CVPR07}
C.~Rother, V.~Kolmogorov, V.~Lempitsky, and M.~Szummer.
\newblock Optimizing binary {MRF}s via extended roof duality.
\newblock In {\em CVPR}, 2007.

\bibitem{shekhovtsov-phd}
A.~Shekhovtsov.
\newblock {\em Exact and Partial Energy Minimization in Computer Vision}.
\newblock {PhD Thesis CTU--CMP--2013--24}, Center for Machine Perception,
  K13133 FEE Czech Technical University in Prague, 2013.

\bibitem{shekhovtsov-11-aux_submodular}
A.~Shekhovtsov and V.~Hlav{\'a}{\v c}.
\newblock On partial opimality by auxiliary submodular problems.
\newblock {\em Control Systems and Computers}, (2), 2011.
\newblock Special issue.

\bibitem{Shekhovtsov-07-binary-TR}
A.~Shekhovtsov, V.~Kolmogorov, P.~Kohli, V.~Hlavac, C.~Rother, and P.~Torr.
\newblock {LP}-relaxation of binarized energy minimization.
\newblock Technical Report {CTU--CMP}--2007--27, Czech Technical University,
  2008.

\bibitem{Schlesinger-76}
M.~Shlezinger.
\newblock Syntactic analysis of two-dimensional visual signals in the presence
  of noise.
\newblock {\em Cybernetics and Systems Analysis}, 4:113--130, 1976.
\newblock See review~\cite{Werner-PAMI07}.

\bibitem{Swoboda-13}
P.~Swoboda, B.~Savchynskyy, J.~Kappes, and C.~Schn\"orr.
\newblock Partial optimality via iterative pruning for the {P}otts model.
\newblock In {\em SSVM}, 2013.

\bibitem{Swoboda-14}
P.~Swoboda, B.~Savchynskyy, J.~H. Kappes, and C.~Schn{\"o}rr.
\newblock Partial optimality by pruning for {MAP}-inference with general
  graphical models.
\newblock In {\em CVPR}, page~8, 2014.

\bibitem{Werner-PAMI07}
T.~Werner.
\newblock A linear programming approach to max-sum problem: {A} review.
\newblock {\em PAMI}, 29(7):1165--1179, July 2007.

\end{thebibliography}
